\documentclass[english,aps,pra,twocolumn]{revtex4-1}

\usepackage[english]{babel}
\usepackage{amsmath}
\usepackage{amsthm}
\usepackage{amsfonts}
\usepackage{amssymb}
\usepackage{amscd}
\usepackage{times}
\usepackage{amsbsy}
\usepackage{hyperref}
\usepackage{graphicx}
\usepackage{booktabs}
\usepackage{array}
\usepackage{subfigure}
\usepackage{enumerate}
\usepackage{url}
\usepackage[nodayofweek]{datetime}
\usepackage[english]{babel}
\usepackage{mathtools}
\usepackage[toc,page]{appendix}
\usepackage{verbatim} 
\usepackage{amsthm} 
\usepackage{enumitem}
\usepackage{enumerate}
\usepackage{bbm} 
\usepackage[normalem]{ulem} 
\usepackage{bm}
\usepackage{mathtools}
\usepackage{color}

\newcommand{\be}{\begin}
\newcommand{\e}{\end}
\newcommand{\beq}{\begin{equation}}
\newcommand{\eeq}{\end{equation}}
\newcommand{\beqs}{\begin{equation*}}
\newcommand{\eeqs}{\end{equation*}}
\renewcommand{\l}{\left}
\renewcommand{\r}{\right}

\renewcommand{\d}{\mathrm{d}} 

\newcommand{\curly}[1]{\mathcal{#1}}

\newcommand{\on}{{\otimes n}}

\newcommand{\om}{\omega}

\newcommand{\Lam}{\Lambda}

\newcommand{\de}{\delta}

\newcommand{\De}{\Delta}

\newcommand{\ket}[1]{|#1\rangle}

\newcommand{\tr}[1]{\textnormal{tr}\left[#1\right]}	
\newcommand{\trr}{\mathrm{tr}}
\newtheorem{thm}{Theorem}
\newtheorem{lm}[thm]{Lemma}

\theoremstyle{definition}

\theoremstyle{remark}

\newcommand{\R}{\mathbb R}

\newcommand{\C}{\mathbb C}

\begin{document}

\title[Approximate cloning and broadcasting]{Information-theoretic limitations on approximate quantum cloning and broadcasting}
\date{\today}

\author{Marius Lemm}
\affiliation{Department of Mathematics, California Institute of Technology, Pasadena, CA 91125}
\author{Mark M.\ Wilde}
\affiliation{Hearne Institute for Theoretical Physics, Department of Physics and Astronomy, Center for Computation and Technology, Louisiana State University, Baton Rouge, Louisiana 70803, USA}

\begin{abstract}
We prove new quantitative limitations on any approximate
simultaneous cloning or broadcasting of mixed states. The results are
based on information-theoretic (entropic) considerations and generalize
the well known no-cloning and no-broadcasting theorems. We also observe and exploit the fact that the universal cloning machine on the symmetric subspace of $n$ qudits and symmetrized partial trace channels are dual to each
other. This duality manifests itself both in the algebraic sense of
adjointness of quantum channels and in the operational sense that a universal
cloning machine can be used as an approximate recovery channel for a
symmetrized partial trace channel and vice versa. The duality extends to give control on the performance of generalized UQCMs on subspaces more general than the symmetric subspace. This gives a way to quantify the usefulness of a-priori information in the context of cloning. For example, we can control the performance of an antisymmetric analogue of the UQCM in recovering from the loss of $n-k$ fermionic particles.
\end{abstract}

\maketitle

A direct consequence of the fundamental principles of quantum theory is that
there does not exist a ``machine'' (unitary map) that can clone an
arbitrary input state \cite{Dieks,Wootters}. This no-cloning theorem and its
generalization to mixed states, the ``no-broadcasting theorem''
\cite{Barnumetal96}, exclude the possibility of making perfect ``quantum
backups'' of a quantum state and are essential for our understanding of
quantum information processing. For instance, since decoherence is
such a formidable obstacle to building a quantum computer and, at the same
time, we cannot use quantum backups to protect quantum
information against this decoherence, considerable effort has been devoted to protecting the stored information by way of 
\emph{quantum error correction}  \cite{Knill,Manda,Shor}.

Given these no-go results, it is natural to ask how well one can do when settling for \emph{approximate} cloning or broadcasting.
Numerous theoretical and experimental works have investigated such ``approximate cloning machines'' (see \cite{Bu,Gisin,PhysRevA.58.1827,Karen,KW99,Lamas,Scarani,Review,Zhu,Chatterjeeetal} and references therein).
These cloning machines can be of great help for \emph{state estimation}. They can also be of great help to an adversary who is eavesdropping on an encrypted communication, and so knowing the limitations of approximate cloning machines is relevant for \emph{quantum key distribution}. 

In this paper, we derive \emph{new quantitative limitations} posed on any approximate cloning/broadcast (defined below) by \emph{quantum information theory}. Our results generalize the standard no-cloning and no-broadcasting results for mixed states, which are recalled below (Theorems 1 and 2). We draw on an approach of Kalev and Hen \cite{KH08},
who introduced the idea of studying no-broadcasting via the fundamental principle of the monotonicity of the quantum relative entropy \cite{Lindblad1975,U77}.
When at least one state is approximately cloned, while the other is approximately broadcast, we derive an inequality which implies rather strong limitations (Theorem~\ref{thm:mainNC}). The result can be understood as a quantitative version of the standard no-cloning theorem. The proof uses only fundamental properties of the relative entropy.
By invoking recent developments linking the monotonicity of relative entropy to recoverability \cite{FR,BLW,SFR,Wilde,Jungeetal,Sutteretal},
we can derive a stronger inequality (Theorem~\ref{thm:mainNCnew}). Under certain circumstances, this stronger inequality provides an \emph{explicit channel} which can be used to \emph{improve the quality} of the original cloning/broadcast (roughly speaking, how close the output is to the input) a posteriori. This cloning/broadcasting-improving channel is nothing
 but the parallel application of the rotation-averaged Petz recovery map \cite{Jungeetal}, highlighting its naturality in this context.

Related results of ours (Theorems~\ref{thm:PT-cloner-recovery} and \ref{thm:cloner-PT-recovery}) compare a given state of $n$ qudits to the maximally mixed state on the (permutation-)symmetric subspace of $n$ qudits. We establish a duality between universal quantum cloning machines (UQCMs) \cite{Bu,Gisin,PhysRevA.58.1827} and symmetrized partial trace channels, in the operational sense that a UQCM can be used as an approximate recovery channel for a symmetrized partial trace channel and vice versa. It is also immediate to observe that these channels are adjoints of each other, up to a constant. A context different from ours, in which a duality between partial trace and universal cloning has been observed, is in quantum data compression \cite{Giulio}.

 As a special case of Theorem~\ref{thm:PT-cloner-recovery}, we recover one of the main results of Werner \cite{PhysRevA.58.1827}, regarding the optimal fidelity for $k\to n$ cloning of tensor-product pure states~$\phi^{\otimes k}$. We also draw an analogy of these results to former results from \cite{PhysRevA.93.062314} regarding photon loss and amplification, the analogy being that cloning is like particle amplification and partial trace like particle loss. 
 
 The methods generalize to subspaces beyond the symmetric subspace: Theorem~\ref{thm:new} controls the performance of an analogue of the UQCM in recovering from a loss of $n-k$ particles when we are given \emph{a priori information} about the states (in the sense that we know on which subspaces they are supported, e.g., because we are working in an irreducible representation of some symmetry group). As an application of this, we obtain an estimate of the performance of an \emph{antisymmetric} analogue of the UQCM for $k\to n$ cloning of fermionic particles. 
 
The methods also yield information-theoretic restrictions for general approximate broadcasts of two mixed states. 



\textit{Background}---The
well known no-cloning theorem for pure states establishes that two pure states can be simultaneously cloned iff they are identical or orthogonal. It is generalized by the following two theorems, a no-cloning theorem for mixed states and a no-broadcasting theorem \cite{Barnumetal96,KH08}. 

Let $\sigma$ be a mixed state on a system $A$. By definition, a (two-fold) \emph{broadcast} of the input state $\sigma$ is a quantum channel $\Lam_{A\to AB}$, such that the output state 
$$
\rho^{\mathrm{out}}_{AB}:=\Lam_{A\to AB}(\sigma_A)
$$
has the identical marginals $\rho^{\mathrm{out}}_{A}=\rho^{\mathrm{out}}_{B}=\sigma$.

A particular broadcast corresponds to the case $\rho^{\mathrm{out}}_{AB}=\sigma_A\otimes \sigma_B$, which is called a \emph{cloning} of the state $\sigma$. We call two mixed states $\sigma_1$ and $\sigma_2$ orthogonal if $\sigma_1\sigma_2=0$.

\be{thm}[No cloning for mixed states, \cite{Barnumetal96,KH08}]
\label{thm:NC}
Two mixed states $\sigma_1,\sigma_2$ can be simultaneously cloned iff they are orthogonal or identical. 
\e{thm}

\be{thm}[No broadcasting, \cite{Barnumetal96}]
\label{thm:UBC}
Two mixed states $\sigma_1,\sigma_2$ can be simultaneously broadcast iff they commute.
\e{thm}

By a ``simultaneous cloning/broadcast,'' we mean that the same choice of $\Lam_{A\to AB}$ is made for broadcasts of $\sigma_1$ and $\sigma_2$.

These results were essentially first proved in \cite{Barnumetal96}, albeit under an additional minor invertibility assumption. Alternative proofs were given in \cite{Lindblad99, Leifer,Barnumetal07, KH08}. Sometimes Theorem~\ref{thm:UBC} is called the ``universal no-broadcasting theorem'' to distinguish it from local no-broadcasting results for multipartite systems \cite{Pianietal}. Quantitative versions of the local no-broad\-casting results for multipartite systems were reviewed very recently by Piani \cite{Piani16} (see also \cite{Chatterjeeetal}). 

No-cloning and no-broadcasting are also closely related to the monogamy property of entanglement via the Choi-Jamiolkowski isomorphism  \cite{Leifer}. 



In this paper, we study limitations on \emph{approximate} cloning/broadcasting, which we define as follows:

\be{defn}[Approximate cloning/broadcast]
Let $\sigma,\tilde\sigma$ be mixed states. An $n$-fold \emph{approximate broadcast} of $\sigma$ is a quantum channel $\Lam_{A\to A_1\cdots A_n}$ such that the output state has the identical marginals $\tilde\sigma$. That is, we consider the situation
\beq
\label{eq:approxdefn}
\rho^{\mathrm{out}}_{A_1}=\cdots=\rho^{\mathrm{out}}_{A_n}=\tilde\sigma,
\eeq
where $\rho^{\mathrm{out}}_{A_1\cdots A_n}:=\Lam(\sigma_A)$. An \emph{approximate cloning} is an approximate broadcast for which $\rho^{\mathrm{out}}_{A_1\cdots A_n}=\tilde\sigma_{A_1}\otimes\cdots\otimes \tilde\sigma_{A_n}$. The main case of interest is $n=2$. 
\e{defn}


Our main results give bounds on (appropriate notions of) distance between $\tilde\sigma_i$ and $\sigma_i$ for $i=1,2$, given any pair of input states $\sigma_1$ and~$\sigma_2$. 

\textit{Conventions}---The notions of approximate cloning / broadcast stated above are direct generalizations of the notions of cloning/broadcasting in the literature related to Theorems~\ref{thm:NC} and~\ref{thm:UBC}. Regarding the input states, these notions are more general than the one used in the cloning machine literature \cite{Scarani}; we allow for the input states to be arbitrary, whereas they are usually pure tensor-power states $\psi^{\on}$ for cloning machines. Our notion of approximate cloning requires the output states to be tensor-product states. Hence, some quantum cloning machines (in particular the universal cloning machine when acting on general input states) are approximate \emph{broadcasts} by the definition given above. 

Let us fix some notation. Given two mixed states~$\rho$ and~$\sigma$, we denote the \emph{relative entropy} of $\rho$ with respect to $\sigma$ by $D(\rho\| \sigma):=\tr{\rho(\log \rho-\log\sigma)}$, where $\log$ is the natural logarithm \cite{U62}. We define the fidelity by $F(\rho,\sigma):=\|\sqrt{\rho}\sqrt{\sigma}\|_1^2 \in [0,1]$ \cite{U73}, where $\|\cdot\|_1$ is the trace norm.

Since all of our bounds involve the relative entropy $D(\sigma_1\|\sigma_2)$ of the input states $\sigma_1$ and $\sigma_2$, they are only informative when $D(\sigma_1\|\sigma_2)<\infty$. This is equivalent to
 $
\ker\sigma_2\subseteq \ker\sigma_1 ,
$
 and we \emph{assume} this in the following for simplicity. We note that if this assumption fails, our results can still be applied by approximating $\sigma_2$ (in trace distance) with $\sigma_2^\varepsilon:=\varepsilon\sigma_1+(1-\varepsilon)\sigma_2$ for $\varepsilon \in (0,1)$, which satisfies $\ker\sigma_2^\varepsilon\subseteq \ker\sigma_1$. 

\textit{Main results}---We will now present our main results. All proofs are rather short and deferred to \cite{LW16supmat}. 

\textit{Restrictions on approximate cloning/broadcasting}---Our first main result concerns limitations if $\sigma_1$ is approximately broadcast $n$-fold while $\sigma_2$ is approximately cloned $n$-fold.

\begin{thm}[Limitations on approximate cloning / broadcasting]
\label{thm:mainNC}
Fix two mixed states $\sigma_{1}$ and $\sigma_{2}$. Let $\Lambda_{A\rightarrow A_{1}\cdots A_{n}}$ be a
quantum channel  such that $n\geq2$
and the two output states $\rho_{i,A_{1}\cdots A_{n}}^{\operatorname{out}%
}:=\Lambda(\sigma_{i,A})$ for $i=1,2$ satisfy%
\beq
\label{eq:approxCB}
\begin{aligned}
\rho_{1,A_{1}}^{\operatorname{out}} &  =\cdots=\rho_{1,A_{n}}%
^{\operatorname{out}}=\tilde{\sigma}_{1},\\
\rho_{2,A_{1}\cdots A_{n}}^{\operatorname{out}} &  =\tilde{\sigma}_{2,A_{1}%
}\otimes\cdots\otimes\tilde{\sigma}_{2,A_{n}},
\end{aligned}.
\eeq
Thus, $\Lambda_{A\rightarrow A_{1}\cdots A_{n}}$ approximately broadcasts
$\sigma_{1,A}$ and approximately clones $\sigma_{2,A}$. Then
\beq
\label{eq:mainNC}
\begin{aligned}
D(\sigma_1\|\sigma_2)-D(\tilde\sigma_1\|\tilde\sigma_2)&\geq  (n-1)D(\tilde\sigma_1\|\tilde\sigma_2)\\
&\geq \frac{n-1}{2}\|\tilde\sigma_1-\tilde\sigma_2\|_1^2.
\end{aligned}
\eeq
\end{thm}

The second inequality in \eqref{eq:mainNC} follows from the quantum Pinsker inequality \cite[Thm.~1.15]{OP93}.

To see that \eqref{eq:mainNC} is indeed restrictive for approximate cloning / broadcasting, let $n=2$ and suppose without loss of generality that $\sigma_1\neq \sigma_2$, so that
$
\de:=\frac{1}{6}\|\sigma_1-\sigma_2\|_1^2>0.
$
We can use the triangle inequality for $\|\cdot\|_1$ and the elementary inequality $2ab\leq a^2+b^2$ on the right-hand side in \eqref{eq:mainNC} to get
\beqs
D(\sigma_1\|\sigma_2)-D(\tilde\sigma_1\|\tilde\sigma_2)
+\frac{\|\sigma_1-\tilde\sigma_1\|_1^2}{2}
+\frac{\|\sigma_2-\tilde\sigma_2\|_1^2}{2}\geq \de.
\eeqs
Since $\sigma_1$ and $\sigma_2$ are fixed, the same is true for $\de>0$. Hence, \emph{for any approximate cloning/broadcasting operation \eqref{eq:approxCB}, at least one of the following three statements must hold}:
\begin{enumerate}
\item $\sigma_1$ is far from $\tilde\sigma_1$ (i.e., the channel acts poorly on the first state),
\item $\sigma_2$ is far from $\tilde\sigma_2$ (i.e., the channel acts poorly on the first state), or
\item there is a large decrease in the distinguishability of the states under the action of the channel, in the sense that $D(\sigma_1\|\sigma_2)-D(\tilde\sigma_1\|\tilde\sigma_2)$ is bounded from below by a constant.
\end{enumerate}

Thus, we have a quantitative version of Theorem \ref{thm:NC} (note that for $\sigma_i=\tilde\sigma_i$ ($i=1,2$), Theorem \ref{thm:mainNCnew} implies $\sigma_1=\sigma_2$). 

As anticipated in the introduction, we can prove a stronger version of Theorem~\ref{thm:mainNC} by invoking recent developments linking monotonicity of the relative entopy to recoverability \cite{FR,BLW,SFR,Wilde,Jungeetal,Sutteretal}.
The stronger version involves an additional non-negative term on the right-hand side in \eqref{eq:mainNC} and it contains an additional integer parameter $m\in\{1,\ldots,n\}$
(the case $m=n$ corresponds to Theorem~\ref{thm:mainNC}; the case $m=1$ is also useful as we explain after the theorem). 


\be{thm}[Stronger version of Theorem~\ref{thm:mainNC}]
\label{thm:mainNCnew}
Under the same assumptions as in Theorem~\ref{thm:mainNC}, for all $m\in\{1,\ldots,n\}$, there exists a recovery channel $\mathcal{R}_{A_{1}\cdots
A_{m}\rightarrow A}^{(m)}$ such that
\begin{multline}
\label{eq:mainNCnew}
D(\sigma_{1}\Vert\sigma_{2})-mD(\tilde{\sigma}_{1}\Vert\tilde{\sigma}_{2}%
)\geq\\
-\log F(\sigma_{1},(\mathcal{R}_{A_{1}\cdots A_{m}\rightarrow A}^{(m)}%
\circ\operatorname{tr}_{A_{m+1}\cdots A_{n}}\circ\Lambda)(\sigma_{1})).
\end{multline}
The recovery channel $\curly{R}^{(m)}\equiv \mathcal{R}_{A_{1}\cdots A_{m}\rightarrow A}^{(m)}$
satisfies the identity $\sigma_{2}=\mathcal{R}%
^{(m)}(\tilde{\sigma}_{2}^{\otimes m}).$ There exists an explicit choice for such an $\curly{R}^{(m)}$  with a formula depending only on $\sigma_2$ and $\Lam$ \cite{Jungeetal, LW16supmat}. 
\e{thm}

One can generalize Theorem~\ref{thm:mainNCnew} to the case of ``$k\to n$ cloning'' \cite{Scarani} where one starts from $k$-fold tensor copies $\sigma_1^{\otimes k}$ and $\sigma_2^{\otimes k}$ and broadcasts the former and clones the 
latter to states on an $n$-fold tensor product; this is Theorem 11 in \cite{LW16supmat}.



To see how the additional remainder term in \eqref{eq:mainNCnew} can be useful, we apply Theorem~\ref{thm:mainNCnew} with $m=1$. 
It implies that there exists a recovery channel $\mathcal{R}%
^{(1)}$ such that%
\beq
\begin{aligned}
D(\sigma_{1}\Vert\sigma_{2})-D(\tilde{\sigma}_{1}\Vert\tilde{\sigma}_{2})  &
\geq-\log F(\sigma_{1},\mathcal{R}^{(1)}(\tilde{\sigma}%
_{1})),\label{eq:local-recovery}\\
\sigma_{2}  &  =\mathcal{R}^{(1)}(\tilde{\sigma}_{2}).
\end{aligned}
\eeq
Now suppose that we are in a situation where the left hand side in \eqref{eq:local-recovery} is less than some $\varepsilon>0$. 
Then, \eqref{eq:local-recovery} implies that $\sigma_{1}    \approx\mathcal{R}^{(1)}(\tilde{\sigma}_{1})$ and $\sigma_{2}   =\mathcal{R}^{(1)}(\tilde{\sigma}_{2})$, where $\approx$ stands for $-\log F(\sigma_{1},\mathcal{R}^{(1)}(\tilde{\sigma}_1))<\varepsilon$. In other words, we can (approximately) recover the input states $\sigma_i$ from the output marginals $\tilde\sigma_i$.
Therefore, in a next step, we can improve the quality of the cloning / broadcasting channel $\Lambda$ by post-composing it with $n$ parallel uses of the local
recovery channel $\mathcal{R}^{(1)}$. Indeed, the \emph{improved cloning channel} $
\Lam_{\mathrm{impr}}:=(\mathcal{R}^{(1)})^{\otimes n}\circ \Lam,
$
has the new output states $\rho_{i,A_{1}\ldots A_n}^{\operatorname{impr}}:=\Lam_{\mathrm{impr}}(\sigma_i)$, $(i=1,2)$ which satisfy
$$
\begin{aligned}
\rho_{1,A_{1}}^{\operatorname{impr}} &  =\cdots=\rho_{1,A_{n}}%
^{\operatorname{impr}}=\mathcal{R}^{(1)}(\tilde{\sigma}_{1})\approx     \sigma_1 , \\
\rho_{2,A_{1}\cdots A_{n}}^{\operatorname{impr}} &  =\sigma_{2,A_{1}%
}\otimes\cdots\otimes\sigma_{2,A_{n}} .
\end{aligned}
$$
Here, $\approx$ again stands for $-\log F(\sigma_{1},\mathcal{R}^{(1)}(\tilde{\sigma}_1))<\varepsilon$.

That is, we have found a strategy to improve the output of the cloning channel $\Lambda$, namely to the output of $\Lam_{\mathrm{impr}}$. 

\textit{Universal cloning machines and symmetrized partial trace channels}---In our next results, we consider a particular example of an approximate broadcasting channel
well known in quantum information theory \cite{PhysRevA.58.1827,KW99,Scarani}, a universal quantum cloning machine (UQCM). We connect the UQCM
to relative entropy and recoverability. 

We recall that the UQCM\ is the
optimal cloner for tensor power pure states, in the sense that the marginal states of its
output have the optimal fidelity with the input state
\cite{PhysRevA.58.1827,KW99}. Let $k$ and $n$ be
integers such that $1\leq k\leq n$. In general, one considers a $k\rightarrow n$
UQCM as acting on $k$ copies $\psi^{\otimes k}$\ of an input pure state $\psi
$ of dimension $d$ (a qudit), which produces an output density operator $\rho^{(n)}$, a state of $n$ qudits. From Werner's work
\cite{PhysRevA.58.1827}, the UQCM is known to be 
\begin{equation}
\mathcal{C}_{k\rightarrow n}(\omega^{(k)})\equiv\frac{d[k]}{d[n]}%
\Pi_{\operatorname{sym}}^{d,n}\left[  \Pi_{\operatorname{sym}}^{d,k}%
\omega^{(k)}\Pi_{\operatorname{sym}}^{d,k}\otimes I^{n-k}\right]
\Pi_{\operatorname{sym}}^{d,n}.
\label{eq:cloning-channel}
\end{equation}
Here $\Pi_{\operatorname{sym}}^{d,n}$ is the projection onto the
(permutation-)symmetric subspace of $(\mathbb{C}^d)^{\otimes n}$, which has dimension $d[n]:=\binom{d+n-1}%
{n}$. We note that $\mathcal{C}_{k\rightarrow n}$ is trace-preserving when acting on the symmetric subspace.

The main results here are Theorems \ref{thm:PT-cloner-recovery} and \ref{thm:cloner-PT-recovery}, which highlight the duality between the UQCM \eqref{eq:cloning-channel} and the following symmetrized partial trace channel
\begin{equation}
\mathcal{P}_{n\rightarrow k}(\cdot)\equiv\Pi_{\operatorname{sym}}%
^{d,k}\operatorname{tr}_{n-k}\!\left[  \Pi_{\operatorname{sym}}^{d,n}(\cdot
)\Pi_{\operatorname{sym}}^{d,n}\right]  \Pi_{\operatorname{sym}}^{d,k},\label{eq:partial-trace-channel}
\end{equation}
In addition to the operational sense of duality between the partial trace channel $\mathcal{P}_{n\rightarrow k}$  and the UQCM $\mathcal{C}_{k\rightarrow n}$ which is established by Theorems \ref{thm:PT-cloner-recovery} and \ref{thm:cloner-PT-recovery}, the two are dual in the sense of quantum channels (up to constant). That is, $\mathcal{P}_{n\rightarrow k}^{\dag}=\left(
d[n]/d[k]\right)  \mathcal{C}_{k\rightarrow n}$.


Our results will quantify the quality of the UQCM for certain tasks in terms of the relative entropy $D(\omega^{(n)}\Vert\pi
_{\operatorname{sym}}^{d,n})$, which is between a general $n$-qudit state $\omega^{(n)}$ and the maximally mixed state
$\pi_{\operatorname{sym}}^{d,n}$ of the symmetric subspace. We consider the maximally mixed state $\pi_{\operatorname{sym}}^{d,n}$ as a natural ``origin'' from which to measure the ``distance'' $D(\omega^{(n)}\Vert\pi
_{\operatorname{sym}}^{d,n})$ since it is a (Haar-)random mixture of tensor-power pure states.

We recall what one obtains from the standard monotonicity of the relative entropy, namely
\begin{equation}
\label{eq:mono2}
D(\omega^{(n)}\Vert\pi_{\operatorname{sym}}^{d,n})\geq D(\mathcal{P}%
_{n\rightarrow k}(\omega^{(n)})\Vert\mathcal{P}_{n\rightarrow k}%
(\pi_{\operatorname{sym}}^{d,n})).
\end{equation}

Our next main result is the following strengthening of the entropy inequality
in \eqref{eq:mono2}:
\be{thm}\label{thm:PT-cloner-recovery}
Let $\omega^{(n)}$ be a state with support in the symmetric subspace of $(\mathbb{C}^d)^{\otimes n}$, let $\pi_{\operatorname{sym}}^{d,n}$ denote the maximally mixed state on this symmetric subspace, let 
$\mathcal{C}_{k\rightarrow n}$ denote the UQCM from \eqref{eq:cloning-channel}, and $\mathcal{P}%
_{n\rightarrow k}$ the symmetrized partial trace channel from
\eqref{eq:partial-trace-channel}. Then
\begin{multline}
D(\omega^{(n)}\Vert\pi_{\operatorname{sym}}^{d,n})\geq D(\mathcal{P}%
_{n\rightarrow k}(\omega^{(n)})\Vert\mathcal{P}_{n\rightarrow k}%
(\pi_{\operatorname{sym}}^{d,n}))\label{eq:refined-partial-trace-cloner}\\
+D(\omega^{(n)}\Vert(\mathcal{C}_{k\rightarrow n}\circ\mathcal{P}%
_{n\rightarrow k})(\omega^{(n)})).
\end{multline}
\e{thm}
The entropy inequality in \eqref{eq:refined-partial-trace-cloner} can be interpreted as follows: The ability
of a $k\rightarrow n$ UQCM to recover an $n$-qubit state $\om^{(n)}$ from the loss of $n-k$ particles is
limited by the decrease of distinguishability between $\omega^{(n)}$ and
$\pi_{\operatorname{sym}}^{d,n}$ under the action of the partial trace
$\mathcal{P}_{n\rightarrow k}$. Thus, a small decrease in relative entropy
(i.e., $D(\omega^{(n)}\Vert\pi_{\operatorname{sym}}^{d,n})-D(\mathcal{P}%
(\omega^{(n)})\Vert\mathcal{P}(\pi_{\operatorname{sym}}^{d,n}))\approx
\varepsilon$) implies that a $k\rightarrow n$ UQCM $\mathcal{C}_{k\rightarrow
n}$ will perform well at recovering $\omega^{(n)}$ from $\mathcal{P}%
_{n\rightarrow k}(\omega^{(n)})$. We can also observe that $\mathcal{C}_{k\rightarrow n}$ is
the Petz recovery map corresponding to the state $\sigma=\pi
_{\operatorname{sym}}^{d,n}$ and channel $\mathcal{N}=\operatorname{tr}_{n-k}$
(as defined in \cite{LW16supmat}). 

As an application of Theorem~\ref{thm:PT-cloner-recovery}, we consider the special case that is most common in the context of quantum cloning \cite{PhysRevA.58.1827,KW99,Scarani}. We set
$\omega^{(n)}=\phi^{\otimes n}$ for a pure state $\phi$. In this case,
\beq
\begin{aligned}
& \!\!\!\!\!\!\!\! D(\phi^{\otimes n}\Vert\pi_{\operatorname{sym}}^{d,n})-D(\mathcal{P}%
_{n\rightarrow k}(\phi^{\otimes n})\Vert\mathcal{P}_{n\rightarrow k}%
(\pi_{\operatorname{sym}}^{d,n})) \\
& =-\log(d[k]/d[n])\geq D(\phi^{\otimes n}\Vert\mathcal{C}_{k\rightarrow n}%
(\phi^{\otimes k})).
\end{aligned}
\eeq
By estimating $D\geq-\log F$, we recover one of the main results of \cite{PhysRevA.58.1827}, which is that the $k\rightarrow n$
UQCM has the following performance when attempting to recover $n$ copies of
$\phi$ from $k$ copies:%
\begin{equation}
\label{eq:Werner}
F(\phi^{\otimes n},\mathcal{C}_{k\rightarrow n}(\phi^{\otimes k}))\geq
d[k]/d[n].
\end{equation}

Given the above duality between the symmetrized partial trace channel and the UQCM, we can
also consider the reverse scenario. 


%
\be{thm}\label{thm:cloner-PT-recovery}
With the same notation as in Theorem~\ref{thm:PT-cloner-recovery}, the following inequality holds
\begin{multline}
D(\omega^{(k)}\Vert\pi_{\operatorname{sym}}^{d,k})\geq D(\mathcal{C}%
_{k\rightarrow n}(\omega^{(k)})\Vert\mathcal{C}_{k\rightarrow n}%
(\pi_{\operatorname{sym}}^{d,k}%
))\label{eq:refined-partial-trace-cloner-reversed}\\
+D(\omega^{(k)}\Vert(\mathcal{P}_{n\rightarrow k}\circ\mathcal{C}%
_{k\rightarrow n})(\omega^{(k)})).
\end{multline}
\e{thm}
This entropy inequality can be seen as dual to that in
\eqref{eq:refined-partial-trace-cloner}, having the following
interpretation:\ if the decrease in distinguishability of $\omega^{(k)}$ and
$\pi_{\operatorname{sym}}^{d,k}$ is small under the action of a
UQCM\ $\mathcal{C}_{k\rightarrow n}$, then the partial trace channel $\mathcal{P}%
_{n\rightarrow k}$ can perform well at recovering the original state
$\omega^{(k)}$ back from the cloned version $\mathcal{C}_{k\rightarrow
n}(\omega^{(k)})$.

There is a striking similarity between the inequalities in
\eqref{eq:refined-partial-trace-cloner} and \eqref{eq:refined-partial-trace-cloner-reversed}
and those from \cite[Sect.~III-A]{PhysRevA.93.062314}, which apply to photonic channels (cf.~\cite{5961814}).
This observation is based on the analogy that cloning is like particle amplification and partial trace is like particle loss and we discuss this further in \cite{LW16supmat}.

\textit{Restrictions on cloning in general subspaces}---We can generalize the discussion in the previous section  to arbitrary subspaces. For $1\leq k\leq n$, let $X_n$ be a $d_{X_n}$-dimensional subspace of $(\C^d)^{\otimes n}$ and let $Y_k$ be a $d_{Y_k}$-dimensional subspace of $(\C^d)^{\otimes k}$. We write $\Pi_{X_n}$, $\Pi_{Y_k}$ for the projections onto these subspaces and $\pi_{X_n}$ and $\pi_{Y_k}$ for the corresponding maximally mixed states. We generalize the definitions in  \eqref{eq:cloning-channel} and \eqref{eq:partial-trace-channel} to
\begin{align}
\label{eq:Cdefn}
\mathcal{C}_{k\rightarrow n}(\cdot)
&\equiv\frac{d_{Y_k}}{d_{X_n}}%
\Pi_{X_n}\left[  \Pi_{Y_k}
(\cdot)\Pi_{Y_k}\otimes I^{n-k}\right]
\Pi_{X_n},\\
\mathcal{P}_{n\rightarrow k}(\cdot)&\equiv\Pi_{Y_k}\operatorname{tr}_{n-k}\!\left[  \Pi_{X_n}(\cdot
)\Pi_{X_n}\right]  \Pi_{Y_k}.
\end{align}
The cloning map $\mathcal{C}_{k\rightarrow n}$ is a direct analogue of the UQCM for the specialized task of recovering a state in the subspace $X_n$ from one in the subspace $Y_k$ (previously, $X_n$ and $Y_k$ were both taken to be the symmetric subspace). By inspection, it is completely positive, and if 
$\operatorname{tr}_{n-k}[ \pi_{X_n} ]  = \pi_{Y_k} $, then it is trace preserving when acting on  any operator with support in $X_n$.
 
The same argument that proves Theorem~\ref{thm:PT-cloner-recovery} then gives
 \be{thm}\label{thm:new}
 Let $\om^{(n)}$ be a state with support in $X_n$, and suppose that $\operatorname{tr}_{n\to k}[\om^{(n)}]$ is supported in $Y_k$.
 Then
\begin{multline}
D(\omega^{(n)}\Vert\pi_{X_n})\geq  D(\mathcal{P}%
_{n\rightarrow k}(\omega^{(n)})\Vert
\pi_{Y_k})\label{eq:DP-form-general}\\
+ D(\omega^{(n)}\Vert(\mathcal{C}_{k\rightarrow n}\circ\mathcal{P}%
_{n\rightarrow k})(\omega^{(n)})).
\end{multline}
\e{thm}

The assumption that $\mathrm{tr}_{n\to k}[\om^{(n)}]$ is supported in $Y_k$ is made for convenience. Without it, the quantity $\mathrm{tr}[\curly{P}_{n\to k}(\om^{(n)})]<1$ would enter in the statement, cf.\ \cite{LW16supmat}. We can obtain a stronger statement under the additional assumption $\operatorname{tr}_{n-k}[ \pi_{X_n} ]  = \pi_{Y_k}$: It implies $\mathcal{P}%
_{n\rightarrow k}(\pi_{X_n})=\pi_{Y_k}$ and that $(\mathcal{C}_{k\rightarrow n}\circ\mathcal{P}%
_{n\rightarrow k})(\omega^{(n)})$ has trace one. 

Theorem \ref{thm:new} controls the performance of the cloning machine $\mathcal{C}_{k\rightarrow n}$ \eqref{eq:Cdefn} in recovering from a loss of $n-k$ particles when \emph{a priori information} about the states is given (in the sense that we know on which subspaces they are supported). To see this, consider, e.g., the case of perfect a priori information when $\dim X_n=1$. Then $D(\omega^{(n)}\Vert\pi_{X_n})=0$ and so \eqref{eq:DP-form-general} implies that the cloning is perfect, $\omega^{(n)}=(\mathcal{C}_{k\rightarrow n}\circ\mathcal{P}%
_{n\rightarrow k})(\omega^{(n)})$.

For non-trivial applications of Theorem \ref{thm:new}, a natural class of subspaces to consider are those associated to irreducible group representations, e.g.\ of the permutation group acting on $(\C^d)^{\otimes n}$. To avoid introducing the representation-theoretic background, we focus here on the case when both $X_n$ and $Y_k$ are taken to be the familiar \emph{antisymmetric} subspace. Physically, the antisymmetric subspace describes fermions and therefore our results have bearing on electronic analogues of the photonic scenarios mentioned above.

 For this part, we let $d\geq n$. An example system for which $d$ can be larger than $n$ is a tight-binding model on $d$ lattice sites, where each site can host a single electron. The antisymmetric subspace $X_n$ has dimension $d_{X_n}= \binom{d}{n}$. The analogue of a tensor-power pure state in the antisymmetric subspace is a \emph{Slater determinant}
 $\vert \Phi_n \rangle \equiv \vert \phi_1 \rangle \wedge \cdots\wedge \vert \phi_n
 \rangle$, where the states $\{ \vert \phi_i \rangle \}_i$ are orthonormal. \cite{LW16supmat} reviews background and how
the marginal $\operatorname{tr}_{n\to k}[\Phi_n]$ is again antisymmetric and has quantum entropy $\log\binom{n}{k}$. Thus,  \eqref{eq:DP-form-general} of Theorem~\ref{thm:new} applies to establish the first inequality of the following:
\beq
\begin{aligned}
\log \binom{d-k}{d-n}
& = -\log\!\l(\binom{d}{k}\cdot \l[\binom{n}{k}\binom{d}{n}\r]^{-1}\r)
\\
&\geq D(\Phi_n\Vert(\mathcal{C}_{k\rightarrow n}\circ\mathcal{P}%
_{n\rightarrow k})(\Phi_n)).
\end{aligned}
\eeq
Using $D\geq -\log F$ again, we conclude that the performance of the antisymmetric cloning machine $\mathcal{C}_{k\rightarrow n}$ in recovering from a loss of $n-k$ fermionic particles is controlled by
\begin{equation}
F(\Phi_n,\, (\mathcal{C}_{k\rightarrow n}\circ\mathcal{P}%
_{n\rightarrow k})(\Phi_n))\geq 
\l[\binom{d-k}{d-n}\r]^{-1}.
\end{equation}
We mention that $(\mathcal{C}_{k\rightarrow n}\circ\mathcal{P}%
_{n\rightarrow k})(\Phi_n)$ has trace one; this follows from the identity $\operatorname{tr}_{n-k}[ \pi_{X_n} ]  = \pi_{Y_k}$ for the antisymmetric subspace (cf.~Lemma 12 in \cite{LW16supmat}). We also mention that the standard symmetric UQCM would produce the zero state in this case and thus yields a (minimal) fidelity of zero.

\textit{General restrictions on approximate broadcasts}---As the introduction mentioned, our methods imply new information-theoretic restrictions on any approximate two-fold broadcast. These are relegated to \cite{LW16supmat}.

\textit{Conclusion}---In this paper, we have proven several entropic inequalities
that pose limitations on the kinds of approximate clonings / broadcasts that are allowed in
quantum information processing. Some of the results generalize the well known no-cloning and no-broadcasting results, restated in Theorems~\ref{thm:NC} and \ref{thm:UBC}. Other results demonstrate how universal cloning machines and partial trace channels are dual to each other, in the sense that one can be used as an approximate recovery channel for the other, with a performance controlled by entropy inequalities. We can also control the performance of an analogue of the UQCM for cloning between any two subspaces. In particular, we obtain bounds on its performance in recovering from a loss of $n-k$ fermionic particles. 

\acknowledgments{
We acknowledge discussions with Sourav Chatterjee and Kaushik Seshadreesan and helpful comments by an anonymous referee. After completing the results of this paper, we learned of the related and concurrent work of 
Marvian and Lloyd \cite{ML16}. We are grateful to them for passing their manuscript along to us.
M.M.W.\ acknowledges support
from the NSF under Award No. 1350397.
}

\onecolumngrid

\pagebreak

\appendix

\section{Monotonicity of the relative entropy and recoverability}

We recall the lower bound from \cite{Jungeetal} on the decrease of the relative entropy for a channel
$\mathcal{N}$ and states $\rho$ and $\sigma$:
\be{thm}[\cite{Jungeetal}]
\label{thm:monotonicity-general}
Let $\beta(t):=\frac{\pi}{2}(1+\cosh(\pi t))^{-1}$. For any two quantum states $\rho,\sigma$ and a channel
$\mathcal{N}$, the following bound holds
\begin{equation}
D(\rho\| \sigma)\geq D(\mathcal{N}(\rho)\| \mathcal{N}(\sigma))
-\int_\R \log F\!\l(\rho, \curly{R}^t_{\mathcal{N},\sigma}(\mathcal{N}(\rho))	\r) \,\d\beta(t), \nonumber
\end{equation}
where the rotated Petz recovery map $\curly{R}^t_{\mathcal{N},\sigma}$ is defined as
$$
\curly{R}^t_{\mathcal{N},\sigma}(\cdot)
:= \sigma^{(1+it)/2} \mathcal{N}^\dag \l[
(\mathcal{N}(\sigma))^{-(1+it)/2} (\cdot)
(\mathcal{N}(\sigma))^{-(1-it)/2}\r]
\sigma^{(1-it)/2},
$$
where $\mathcal{N}^\dag$ is the completely positive, unital adjoint of the channel $\mathcal{N}$. Every rotated Petz recovery map perfectly recovers $\sigma$
from $\mathcal{N}(\sigma)$:
$$
\curly{R}^t_{\mathcal{N},\sigma}(\mathcal{N}(\sigma))
=\sigma.
$$
\e{thm}

In the special case when the applied quantum channel is the partial trace, the inequality becomes as follows:

\be{thm}[\cite{Jungeetal}]
\label{thm:monotonicity}
Let $\beta(t):=\frac{\pi}{2}(1+\cosh(\pi t))^{-1}$. For any two quantum states $\rho_{AB},\sigma_{AB}$, we have
\begin{equation}
D(\rho_{AB}\| \sigma_{AB})\geq
D(\rho_B\| \sigma_B)
-\int_\R \log F\!\l(\rho_{AB}, \curly{R}^t_{A,\sigma}(\rho_B)	\r) \,\d\beta(t), \nonumber
\end{equation}
where the rotated Petz recovery map $\curly{R}^t_{A,X}$ is defined in \eqref{eq:rPetz}.
\e{thm}


\section{A generalization of Theorem~\ref{thm:mainNCnew} to $k$ to $n$ cloning}
\be{thm}
\label{thm:mainNCnewgeneral}
Consider the more general situation in which we begin with $k\leq n$ tensor-product copies of the state $\sigma_i$ for $i\in \{1,2\}$, and suppose that the channel $\Lambda_{A_1 \cdots A_k \to A_1 \cdots A_n}$ approximately broadcasts $\sigma_1$, in the sense that
$$
\operatorname{tr}_{A_1 \cdots A_n \backslash A_j}[\Lambda_{A_1 \cdots A_k \to A_1 \cdots A_n}(\sigma_1^{\otimes k})] = \tilde\sigma_1,
$$
and approximately clones $\sigma_2$, in the sense that
$$
\Lambda_{A_1 \cdots A_k \to A_1 \cdots A_n}(\sigma_2^{\otimes k}) = \tilde\sigma_2^{\otimes n}.
$$
Then, for every 
$m\in\{1,\ldots,n\}$, there exists a recovery channel $\mathcal{R}_{A_{1}\cdots
A_{m}\rightarrow A_1 \cdots A_k}^{(m,k)}$ such that%
\beq
kD(\sigma_{1}\Vert\sigma_{2})-mD(\tilde{\sigma}_{1}\Vert\tilde{\sigma}_{2}%
)\geq
-\log F(\sigma_{1},(\mathcal{R}_{A_{1}\cdots A_{m}\rightarrow A_1 \cdots A_k}^{(m,k)}%
\circ\operatorname{tr}_{A_{m+1}\cdots A_{n}}\circ\Lambda)(\sigma_{1}^{\otimes k})),\nonumber
\eeq
and the recovery channel $\mathcal{R}_{A_{1}\cdots A_{m}\rightarrow A_1 \cdots A_k}^{(m,k)}$
satisfies $$\sigma_{2}^{\otimes k}=\mathcal{R}_{A_{1}\cdots A_{m}\rightarrow A_1 \cdots A_k}^{(m,k)}(\tilde{\sigma}_{2}^{\otimes m}).$$
\e{thm}

This can be proved by the same method as for Theorem~\ref{thm:mainNCnew} (see below).

\section{On photon amplification and loss}
Here we discuss the analogy between \eqref{eq:refined-partial-trace-cloner} and \eqref{eq:refined-partial-trace-cloner-reversed}
and the inequalities from Section III-A of \cite{PhysRevA.93.062314}. 
The partial trace channel is like particle loss, which for photons is represented
by a pure-loss channel $\mathcal{L}_{\eta}$ with transmissivity $\eta
\in\left[  0,1\right]  $. Furthermore, a UQCM is like particle
amplification, which for bosons is represented by an amplifier channel
$\mathcal{A}_{G}$\ of gain $G\geq1$. Let $\theta_{E}$ denote a thermal state
of mean photon number $E\geq0$, and let $\rho$ denote a state of the same
energy $E$. A slight rewriting of the inequalities from Section III-A of \cite{PhysRevA.93.062314}, given below, results
in the following:%
\begin{align}
D(\rho\Vert\theta_{E})& \gtrsim D(\mathcal{L}_{\eta}(\rho)\Vert\mathcal{L}%
_{\eta}(\theta_{E}))\nonumber\\
& \qquad\qquad+D(\rho\Vert(\mathcal{A}_{1/\eta}\circ\mathcal{L}_{\eta})(\rho
)),\label{eq:loss-ch}\\
D(\rho\Vert\theta_{E})&\geq D(\mathcal{A}_{G}(\rho)\Vert\mathcal{A}%
_{G}(\theta_{E}))\nonumber\\
& \qquad\qquad+D(\rho\Vert(\mathcal{L}_{1/G}\circ\mathcal{A}_{G})(\rho)),\label{eq:amp-ch}%
\end{align}
where the symbol $\gtrsim$ indicates that the entropy inequality holds up to a
term with magnitude no larger than $\log(1/\eta)$ and which approaches zero as $E\rightarrow\infty$. So
we see that \eqref{eq:loss-ch} is analogous to
\eqref{eq:refined-partial-trace-cloner}:\ under a particle loss $\mathcal{L}%
_{\eta}$, we can apply a particle amplification procedure $\mathcal{A}%
_{1/\eta}$ to try and recover the lost particles, with a performance controlled
by \eqref{eq:loss-ch}. Similarly, \eqref{eq:amp-ch} is analogous to
\eqref{eq:refined-partial-trace-cloner-reversed}: under a particle
amplification $\mathcal{A}_{G}$, we can apply a particle loss channel
$\mathcal{L}_{1/G}$ to try and recover the original state, with a performance
controlled by \eqref{eq:amp-ch}. Observe that the parameters specifying the recovery
channels are directly related to the parameters of the original channels, just
as is the case in
\eqref{eq:refined-partial-trace-cloner} and \eqref{eq:refined-partial-trace-cloner-reversed}.
Note that an explicit connection between cloning and amplifier channels was
established in \cite{5961814}, and our result serves to complement that connection.

\be{proof}[Proof of \eqref{eq:loss-ch} and 
\eqref{eq:amp-ch}]
A proof of \eqref{eq:loss-ch} is as follows. The Hamiltonian here is $a^{\dag
}a$, which is the photon number operator. Let $\rho$ be a state of energy $E$,
and let $\theta_{E}$ be a thermal state of energy $E$ (i.e., $\left\langle
a^{\dag}a\right\rangle _{\rho}=\left\langle a^{\dag}a\right\rangle
_{\theta_{E}}=E$). Under the action of a pure-loss channel $\mathcal{L}_{\eta
}$, the energies of $\mathcal{L}_{\eta}(\rho)$ and $\mathcal{L}_{\eta}%
(\theta_{E})$ are equal to $\eta E$, and we also find that $\mathcal{L}_{\eta
}(\theta_{E})=\theta_{\eta E}$. Furthermore, a standard calculation gives that
$-\operatorname{tr}[\rho\log\theta_{E}]=H(\theta_{E})=g(E):=\left(
E+1\right)  \log\left(  E+1\right)  -E\log E$. Putting this together, we find
that%
\begin{align}
D(\rho\Vert\theta_{E})-D(\mathcal{L}_{\eta}(\rho)\Vert\mathcal{L}_{\eta
}(\theta_{E}))
 &=H(\mathcal{L}_{\eta}(\rho))-H(\rho)+g(E)-g(\eta E)\\
& \geq D(\rho\Vert(\mathcal{A}_{1/\eta}\circ\mathcal{L}_{\eta})(\rho
))-\log(1/\eta)+g(E)-g(\eta E).
\end{align}
The first equality is a rewriting using what we mentioned above and the
inequality follows from Section III-A of \cite{PhysRevA.93.062314}. When
$E=0$, $g(E)-g(\eta E)=0$ also. As $E$ gets larger, $g(E)-g(\eta E)$ is
monotone increasing and reaches its maximum of $\log(1/\eta)$ as
$E\rightarrow\infty$.

The other inequality in \eqref{eq:amp-ch} for an amplifier channel follows similarly. Under the
action of an amplifier channel $\mathcal{A}_{G}$, the energies of
$\mathcal{A}_{G}(\rho)$ and $\mathcal{A}_{G}(\theta_{E})$ are $GE$. We also
find that $\mathcal{A}_{G}(\theta_{E})=\theta_{GE}$. Proceeding as above, we
find that%
\begin{align}
 D(\rho\Vert\theta_{E})-D(\mathcal{A}_{G}(\rho)\Vert\mathcal{A}_{G}%
(\theta_{E}))
& =H(\mathcal{A}_{G}(\rho))-H(\rho)+g(E)-g(GE)\\
& \geq D(\rho\Vert(\mathcal{L}_{1/G}\circ\mathcal{A}_{G})(\rho))+\log
G-\left[  g(GE)-g(E)\right]  \\
& \geq D(\rho\Vert(\mathcal{L}_{1/G}\circ\mathcal{A}_{G})(\rho)).
\end{align}
The first equality is a rewriting and the inequality follows from Section
III-A of \cite{PhysRevA.93.062314}. The last inequality follows because
$g(GE)-g(E)=0$ at $E=0$, and it is monotone increasing as a function of $E$,
reaching its maximum value of $\log G$ as $E\rightarrow\infty$.
\e{proof}

\section{Proofs of the main results}

\be{proof}[Proof of Theorems \ref{thm:mainNC} and \ref{thm:mainNCnew}]
Theorem~\ref{thm:mainNC} follows from the $m=n$ case of Theorem~\ref{thm:mainNCnew}. Hence, it suffices to prove Theorem~\ref{thm:mainNCnew}. 
We start by noting the following general inequality
holding for states $\omega$ and $\tau$, a channel $\mathcal{N}$, and a
recovery channel $\mathcal{R}$:%
\begin{align}
D(\omega\Vert\tau)-D(\mathcal{N}(\omega)\Vert\mathcal{N}(\tau)) &  \geq-\log
F(\omega,(\mathcal{R}\circ\mathcal{N})(\omega)),\label{eq:junge-et-al-1}\\
\tau &  =(\mathcal{R}\circ\mathcal{N})(\tau),\label{eq:junge-et-al-2}%
\end{align}
which is a consequence of convexity of $-\log$ and the fidelity applied to Theorem~\ref{thm:monotonicity-general}, taking 
\beq
\label{eq:Rdefn}
\mathcal{R} :=
\int_{\mathbb{R}} \curly{R}^t_{\mathcal{N},\tau} \d \beta(t)
\eeq
with $\curly{R}^t_{\mathcal{N},\tau}$ as in Theorem~\ref{thm:monotonicity-general}.
To get the inequality, we take $\omega=\sigma_{1}$, $\tau=\sigma_{2}$,
and $\mathcal{N}=\operatorname{tr}_{A_{m+1}\cdots A_{n}}\circ\Lambda$.
This then gives the inequality%
\begin{equation}
D(\sigma_{1}\Vert\sigma_{2})-D((\operatorname{tr}_{A_{m+1}\cdots A_{n}}%
\circ\Lambda)(\sigma_{1})\Vert(\operatorname{tr}_{A_{m+1}\cdots A_{n}}%
\circ\Lambda)(\sigma_{2}))\geq-\log F(\sigma_{1},(\mathcal{R}_{A_{1}\cdots
A_{n}\rightarrow A}^{(m)}\circ\operatorname{tr}_{A_{m+1}\cdots A_{n}}%
\circ\Lambda)(\sigma_{1})),
\end{equation}
where the recovery channel $\mathcal{R}_{A_{1}\cdots A_{n}\rightarrow A}%
^{(m)}$ satisfies%
\begin{align}
\sigma_{2}  & =(\mathcal{R}_{A_{1}\cdots A_{n}\rightarrow A}^{(m)}%
\circ\operatorname{tr}_{A_{m+1}\cdots A_{n}}\circ\Lambda)(\sigma_{2})
 =\mathcal{R}_{A_{1}\cdots A_{n}\rightarrow A}^{(m)}(\tilde{\sigma}%
_{2}^{\otimes m}).
\end{align}
So then we prove that $-D((\operatorname{tr}_{A_{m+1}\cdots A_{n}}\circ
\Lambda)(\sigma_{1})\Vert(\operatorname{tr}_{A_{m+1}\cdots A_{n}}\circ
\Lambda)(\sigma_{2}))\leq-mD(\tilde{\sigma}_{1}\Vert\tilde{\sigma}_{2})$. We apply $\log(X\otimes Y)=\log X\otimes I+I\otimes  \log Y$ and set $H(X):=-\tr{X\log X}$ to get
\begin{align}
& \!\!\!\!\!\!\!\!\!\! -D((\operatorname{tr}_{A_{m+1}\cdots A_{n}}\circ\Lambda)(\sigma_{1}%
)\Vert(\operatorname{tr}_{A_{m+1}\cdots A_{n}}\circ\Lambda)(\sigma_{2})) \nonumber\\
&  =-D(\rho_{1,A_{1}\cdots A_{m}}^{\operatorname{out}}\Vert\tilde{\sigma
}_{2,A_{1}}\otimes\cdots\otimes\tilde{\sigma}_{2,A_{m}})\\
&  =H(\rho_{1,A_{1}\cdots A_{m}}^{\operatorname{out}})+\operatorname{tr}%
[\rho_{1,A_{1}\cdots A_{m}}^{\operatorname{out}}\log(\tilde{\sigma}_{2,A_{1}%
}\otimes\cdots\otimes\tilde{\sigma}_{2,A_{m}})]\\
&  =H(\rho_{1,A_{1}\cdots A_{m}}^{\operatorname{out}})+\sum_{k=1}%
^{m}\operatorname{tr}[\rho_{1,A_{1}\cdots A_{m}}^{\operatorname{out}}%
(I_{A_{1}\cdots A_{m}\backslash A_{k}}\otimes\log(\tilde{\sigma}_{2,A_{k}}))]
\end{align}
Recall our assumption from \eqref{eq:approxCB} that the channel broadcasts $\sigma_1$ to $\tilde\sigma_1$. It gives
\begin{align}
&\!\!\!\!\!\!\!\!H(\rho_{1,A_{1}\cdots A_{m}}^{\operatorname{out}})+\sum_{k=1}%
^{m}\operatorname{tr}[\rho_{1,A_{1}\cdots A_{m}}^{\operatorname{out}}
(I_{A_{1}\cdots A_{m}\backslash A_{k}}\otimes\log(\tilde{\sigma}_{2,A_{k}}))] \nonumber \\
&  =H(\rho_{1,A_{1}\cdots A_{m}}^{\operatorname{out}})+\sum_{k=1}%
^{m}\operatorname{tr}[\tilde{\sigma}_{1}\log\tilde{\sigma}_{2}]\\
&  \leq\sum_{k=1}^{m}\left[  H(\rho_{1,A_{k}}^{\operatorname{out}%
})+\operatorname{tr}[\tilde{\sigma}_{1}\log\tilde{\sigma}_{2}]\right]  \\
&  =-mD(\tilde{\sigma}_{1}\Vert\tilde{\sigma}_{2}).
\end{align}
In the second-to-last step, we used the subadditivity of the entropy $H$ and again \eqref{eq:approxCB}. 
\e{proof}

\be{proof}[Proof of Theorem~\ref{thm:PT-cloner-recovery}.]

We observe that
$\pi_{\operatorname{sym}}^{d,k}=\operatorname{tr}_{n-k}[\pi
_{\operatorname{sym}}^{d,n}]$ which follows easily from the representation
$\pi_{\operatorname{sym}}^{d,n}=\int d\psi\ \psi^{\otimes n}$ \cite{H13}, the
integral being with respect to the Haar probability measure over pure states
$\psi$.

A proof of \eqref{eq:refined-partial-trace-cloner}\ then follows from a few key
steps:%
\begin{align}
 & D(\omega^{(n)}\Vert\pi_{\operatorname{sym}}^{d,n})-D(\mathcal{P}%
_{n\rightarrow k}(\omega^{(n)})\Vert\mathcal{P}_{n\rightarrow k}%
(\pi_{\operatorname{sym}}^{d,n}))\nonumber \\
  =&-H(\omega^{(n)})-\operatorname{tr}[\omega^{(n)}\log\pi_{\operatorname{sym}%
}^{d,n}]+H(\mathcal{P}_{n\rightarrow k}(\omega^{(n)}))+\operatorname{tr}%
[\mathcal{P}
_{n\rightarrow k}(\omega^{(n)})\log\pi_{\operatorname{sym}}^{d,k}]\nonumber\\
 =& H(\mathcal{P}_{n\rightarrow k}(\omega^{(n)}))-H(\omega^{(n)}%
)-\log(d[k]/d[n])\nonumber\\
 \geq& D(\omega^{(n)}\Vert(\mathcal{P}_{n\rightarrow k}^{\dag}\circ
\mathcal{P}_{n\rightarrow k})(\omega^{(n)}))-\log(d[k]/d[n])\nonumber\\
 =& D(\omega^{(n)}\Vert(\mathcal{C}_{k\rightarrow n}\circ\mathcal{P}%
_{n\rightarrow k})(\omega^{(n)}%
)).\label{eq:proof-cloning-partial-trace-recovery}%
\end{align}
The first equality holds by definition of quantum relative entropy and in the second equality we used the fact that
 $\mathrm{tr}[\mathcal{P}_{n\rightarrow k}(\omega^{(n)})]=\mathrm{tr}[\mathrm{tr}_{n\to k}(\omega^{(n)})]=\mathrm{tr}[\om^{(n)}]=1$, wherein the first step holds because $\mathrm{tr}_{n\to k}[\om^{(n)}]$ is supported in the symmetric subspace. The inequality above is a consequence of \cite[Thm.~1]{PhysRevA.93.062314}
which states that%
\begin{equation}
H(\mathcal{N}(\rho))-H(\rho)\geq D(\rho\Vert(\mathcal{N}^{\dag}\circ
\mathcal{N})(\rho))
\end{equation}
for any state $\rho$ and positive, trace-preserving map $\mathcal{N}$. (We remark that $\curly{P}_{n\to k}$ is indeed trace-preserving when considered as a map on states supported on the symmetric subspace.) The last equality in \eqref{eq:proof-cloning-partial-trace-recovery} follows
from the property of relative entropy that $D(\xi\Vert
\tau)-\log c=D(\xi\Vert c\tau)$ for states $\xi,\tau$ and $c>0$.
\e{proof}

Essentially the same argument, with minor modifications, also proves Theorems\ref{thm:cloner-PT-recovery} and \ref{thm:new}. For the former, we use the facts that
$\mathcal{C}_{k\rightarrow n}(\pi_{\operatorname{sym}}^{d,k})=\pi
_{\operatorname{sym}}^{d,n}$ and that $\mathcal{C}_{k\rightarrow n}$ is trace-preserving when acting on states supported in the symmmetric subspace.
For Theorem \ref{thm:new}, we use the assumption that $\mathrm{tr}_{n\to k}[\om^{(n)}]$ is supported in $Y_k$ to get $\mathrm{tr}[\mathcal{P}%
_{n\rightarrow k}(\omega^{(n)})]=1$. The details are left to the reader.

We close this proof section with a remark on a so-far implicit assumption.

\be{rmk}[Non-identical marginals case]
Some of our results, Theorems~\ref{thm:mainNC}, \ref{thm:mainNCnew} and \ref{thm:mainBC} (see below), apply to approximate clonings/broadcasts in the sense of Definition 3. That is, we always assume that the marginals of the output state are identical, i.e.
\beq
\label{eq:non-identical}
\rho^{\mathrm{out}}_{i,A_1}=\ldots=\rho^{\mathrm{out}}_{i,A_n}=\tilde\sigma_i, \qquad (i=1,2).
\eeq
We make this assumption for two reasons: (a) It simplifies the bounds in our main results and (b) we believe that it is a natural assumption for approximate cloning/broadcasting. However, the methods apply more generally and they also yield limitations on approximate clonings/broadcasts when \eqref{eq:non-identical} is not satisfied.
\e{rmk}

\section{The maximally mixed state on the antisymmetric subspace}

The following lemma allows us to conclude that the stronger form of Theorem~\ref{thm:new} applies when considering cloning maps for the antisymmetric subspace.

\be{lm}\label{lm:pi}
Let $\curly{H}_n$ denote the antisymmetric subspace of $n$ qudits and let $\pi_n$ denote the maximally mixed state on $\curly{H}_n$. Then
$$
\pi_k=\trr_{n\to k}[\pi_n].
$$
\e{lm}

\be{proof}[Proof of Lemma \ref{lm:pi}]
The operator $\trr_{n\to k}[\pi_n]$ is supported on $\curly{H}_k$. It also commutes with all unitaries $U_k$ on $\curly{H}_k$. Indeed, by properties of the partial trace and the fact that $\pi_n$ commutes with all unitaries on $\curly{H}_n$,
$$
U_k\trr_{n\to k}[\pi_n]=\trr_{n\to k}[(U_k\otimes I_{\curly{H}_{n-k}})\pi_n]=\trr_{n\to k}[\pi_n (U_k\otimes I_{\curly{H}_{n-k}})]=\trr_{n\to k}[\pi_n]U_k.
$$
Since it commutes with all unitaries, $\trr_{n\to k}[\pi_n]$ is proportional to $I_{\curly{H}_k}$. Since 
$$
\trr_{\curly{H}_k}[\trr_{n\to k}[\pi_n]]=\trr_{\curly{H}_n}[\pi_n]=1,
$$ the proportionality constant must be $1/\mathrm{\dim}{\curly{H}_k}=1/\binom{d}{k}$. This proves the lemma.
\e{proof} 

\section{Reductions of Slater determinants and their quantum entropy}

\label{sec:reductions-Slater-dets}

Here we prove the fact that the quantum entropy of the marginal
$\mathrm{tr}_{n\rightarrow k}[\Phi_{n}]$ is $\log\binom{n}{k}$ when $\Phi_{n}$
is a Slater determinant. We can conclude this directly from the expression
\eqref{eq:marginal} for the marginal derived below.

Before beginning, let us suppose that $\{|\phi_{j}\rangle\}_{j=1}^{d}$ is an
orthonormal basis for a $d$-dimensional Hilbert space $\mathcal{H}$. Letting
$d\geq n$, a Slater determinant state $\Phi_{n}$\ corresponding to this basis
and a subset $\{1,\ldots,n\}$ is as follows:%
\begin{align}
|\Phi_{n}\rangle &  :=|\phi_{1}\rangle\wedge\cdots\wedge|\phi_{n}\rangle\\
&  :=\frac{1}{\sqrt{n!}}\sum_{\pi\in S_{n}}\mathrm{sgn}(\pi)|\phi_{\pi
(1)}\rangle\otimes\cdots\otimes|\phi_{\pi(n)}\rangle,
\end{align}
where $S_{n}$ is the set of all permutations of $\{1,\ldots,n\}$ and $\mathrm{sgn}%
(\pi)$ denotes its signum. Note that we chose the subset $\left\{
1,\ldots,n\right\}  $ of $\{1,\ldots,d\}$, but without loss of generality we
could have chosen an arbitrary one.

The formula \eqref{eq:marginal} below is presumably well known. We include an
elementary, but slightly tedious, proof for completeness.

\begin{lm}
[Marginal of a Slater determinant]Let $d \geq n$ and $|\Phi_{n}\rangle=|\phi_{1}\rangle
\wedge\cdots\wedge|\phi_{n}\rangle$, with $\{|\phi_{j}\rangle\}_{j=1}^{d}$ an orthonormal basis. A $k$-set
$A_{k}$ is a subset of $\{1,\ldots,n\}$ consisting of exactly $k$ elements.
For any $k$-set $A_{k}=\{i_{1},\ldots,i_{k}\}$, we define
\begin{equation}
|\Phi_{A_{k}}\rangle\langle\Phi_{A_k}|:=(|\phi_{i_{1}}\rangle\wedge\cdots\wedge|\phi_{i_{k}%
}\rangle)(\langle\phi_{i_{1}}|\wedge\cdots\wedge|\phi_{i_{k}%
}|).
\end{equation}
Then%
\begin{equation}
\mathrm{tr}_{n\rightarrow k}[|\Phi_{n}\rangle\langle\Phi_{n}|]=\frac{1}%
{\binom{n}{k}}\sum_{A_{k}\ k\mathrm{-set}}|\Phi_{A_{k}}\rangle\langle
\Phi_{A_{k}}|. \label{eq:marginal}%
\end{equation}
The orthonormality of the states $\{|\Phi_{A_{k}}\rangle\}$ for fixed $k$ then implies that $H(\mathrm{tr}_{n\rightarrow k}|\Phi_{n}\rangle\langle
\Phi_{n}|)=\log\binom{n}{k}$, where $H(\rho)=-\mathrm{tr}[\rho\log\rho]$ is the
quantum entropy.
\end{lm}

\begin{proof}
By definition of the wedge product, we can write $|\Phi_{n}\rangle\langle
\Phi_{n}|$ as%
\begin{equation}
|\Phi_{n}\rangle\langle\Phi_{n}|=\frac{1}{n!}\sum_{\pi,\sigma\in S_{n}%
}\mathrm{sgn}(\pi)\mathrm{sgn}(\sigma)|\phi_{\pi(1)}\rangle\langle\phi
_{\sigma(1)}|\otimes\cdots\otimes|\phi_{\pi(n)}\rangle\langle\phi_{\sigma
(n)}|.
\end{equation}
Taking the partial trace over the last $n-k$ systems yields the following:%
\begin{align}
&  \mathrm{tr}_{n\rightarrow k}[|\Phi_{n}\rangle\langle\Phi_{n}|]\nonumber\\
&  =\frac{1}{n!}\sum_{\pi,\sigma\in S_{n}}\mathrm{sgn}(\pi)\mathrm{sgn}%
(\sigma)|\phi_{\pi(1)}\rangle\langle\phi_{\sigma(1)}|\otimes\cdots\otimes
|\phi_{\pi(k)}\rangle\langle\phi_{\sigma(k)}|\ \langle\phi_{\pi(k+1)}%
|\phi_{\sigma(k+1)}\rangle\cdots\langle\phi_{\pi(n)}|\phi_{\sigma(k)}%
\rangle\label{eq:permutation}\\
&  =\frac{1}{n!}\sum_{\pi,\sigma\in S_{n}}\mathrm{sgn}(\pi)\mathrm{sgn}%
(\sigma)|\phi_{\pi(1)}\rangle\langle\phi_{\sigma(1)}|\otimes\cdots\otimes
|\phi_{\pi(k)}\rangle\langle\phi_{\sigma(k)}|\ \delta_{\pi(k+1),\sigma
(k+1)}\cdots\delta_{\pi(n),\sigma(n)}.
\end{align}
In the second equality, we used orthonormality. The product of delta functions
implies that we only need to consider permutations $\pi$ and $\sigma$ which
agree on $\{k+1,\ldots,n\}$.

To exploit this, we partition the permutations according to which $k$-set
$A_{k}$ features as the image of $\{1,\ldots,k\}$. More precisely, given a $k
$-set $A_{k}$, we define
\begin{equation}
S_{n}(A_{k}):=\left\{  \pi\in S_{n}\;:\;\pi(\{1,\ldots,k\})=A_{k}\right\}  .
\end{equation}
There is a more useful, kind of affine representation of the elements of
$S_{n}(A_{k})$ as tuples in $S_{k}\times S_{n-k}$ composed with a fixed
bijection $f_{A_{k}}\in S_{n}(A_{k})$. For definiteness, we define $f_{A_{k}}
$ to be the unique bijection in $S_{n}(A_{k})$ which preserves ordering. Then
\begin{equation}
\pi\in S_{n}(A_{k})\Longleftrightarrow\pi=f_{A_{k}}\circ(\pi^{k},\pi
^{n-k}),\quad\text{for some }\pi^{k}\in S_{k},\,\pi^{n-k}\in S_{n-k}.
\label{eq:represent}%
\end{equation}
Here we wrote $(\pi^{k},\pi^{n-k})$ for the permutation that is obtained by
applying $\pi^{k}$ to the first $k$ variables and $\pi^{n-k}$ to the last
$n-k$ variables.

This way of bookkeeping permutations is convenient in \eqref{eq:permutation}
above. Using this representation and the identity \eqref{eq:identity} below, we find that%
\begin{align}
&  \mathrm{tr}_{n\rightarrow k}[|\Phi_{n}\rangle\langle\Phi_{n}|]\nonumber\\
&  =\frac{1}{n!}\sum_{A_{k}\ k\mathrm{-set}}\sum_{\substack{\pi,\sigma\in
S_{n}(A_{k}); \\\pi^{n-k}=\sigma^{n-k}}}\mathrm{sgn}(\pi)\mathrm{sgn}%
(\sigma)|\phi_{\pi(1)}\rangle\langle\phi_{\sigma(1)}|\otimes\cdots\otimes
|\phi_{\pi(k)}\rangle\langle\phi_{\sigma(k)}|\\
&  =\frac{1}{n!}\sum_{A_{k}\ k\mathrm{-set}}\sum_{\substack{\pi,\sigma\in
S_{n}(A_{k}); \\\pi^{n-k}=\sigma^{n-k}}}\mathrm{sgn}(\pi^{k})\mathrm{sgn}%
(\sigma^{k})|\phi_{\pi(1)}\rangle\langle\phi_{\sigma(1)}|\otimes\cdots
\otimes|\phi_{\pi(k)}\rangle\langle\phi_{\sigma(k)}|\\
&  =\frac{(n-k)!}{n!}\sum_{A_{k}\ k\mathrm{-set}}\sum_{\pi^{k},\sigma^{k}\in
S_{k}}\mathrm{sgn}(\pi^{k})\mathrm{sgn}(\sigma^{k}) |\phi_{(f_{A_{k}}\circ
\pi^{k})(1)}\rangle\langle\phi_{(f_{A_{k}}\circ\sigma^{k})(1)}|\otimes
\cdots\otimes|\phi_{(f_{A_{k}}\circ\pi^{k})(k)}\rangle\langle\phi_{(f_{A_{k}%
}\circ\sigma^{k})(k)}|\label{eq:aboveperm}.
\end{align}
We used the following identity:%
\begin{equation}\label{eq:identity}
\mathrm{sgn}(\pi)\mathrm{sgn}(\sigma)=\mathrm{sgn}(\pi^{k})\mathrm{sgn}%
(\sigma^{k}).
\end{equation}
This  is a consequence of the fact that $\mathrm{sgn}$ is a group
homomorphism, i.e., that $\mathrm{sgn}(\sigma_{1}\circ\sigma_{2}%
)=\mathrm{sgn}(\sigma_{1})\mathrm{sgn}(\sigma_{2})$ holds for any two
permutations $\sigma_{1}$ and $\sigma_{2}$. Indeed, we have%
$$
\begin{aligned}
\mathrm{sgn}(\pi)\mathrm{sgn}(\sigma) &  =(\mathrm{sgn}(f_{A_{k}}%
))^{2}\mathrm{sgn}((\pi^{k},\pi^{n-k}))\mathrm{sgn}((\sigma^{k},\sigma
^{n-k}))\\
&  =\mathrm{sgn}((\pi^{k},\pi^{n-k}))\mathrm{sgn}((\sigma^{k},\pi^{n-k}))\\
&  =\mathrm{sgn}((\pi^{k},I_{n-k})\circ(I_{k},\pi^{n-k}))\mathrm{sgn}%
((\sigma^{k},I_{n-k})\circ(I_{k},\pi^{n-k}))\\
&  =\mathrm{sgn}(\pi^{k})\mathrm{sgn}(\sigma^{k}).
\end{aligned}
$$
This proves \eqref{eq:identity}. We now return to \eqref{eq:aboveperm} to conclude the proof of \eqref{eq:marginal}. We observe that 
$$
\mathrm{Perm}(A_{k})=\left\{
f_{A_{k}}\circ\pi^{k}\circ f_{A_{k}}^{-1}\;:\;\pi^{k}\in S_{k}\right\}.
$$
To exploit this, we order each $k$-set $A_k=\{i_1,\ldots,i_k\}$ with $i_1<\cdots<i_k$. Then, by definition, $f_{A_k}(j)=i_j$ for all $1\leq j\leq k$. From this, we find that
$$
f_{A_k}\circ\pi^k(j)=f_{A_k}\circ\pi^k\circ f_{A_k}^{-1}(i_j)=:\tilde\pi^k(i_j)
$$ 
produces a permutation $\tilde\pi^k\in \mathrm{Perm}(A_{k})$. We use this observation to relabel the sum in \eqref{eq:aboveperm}; and we also use  the identity $\mathrm{sgn}(\pi^{k})\mathrm{sgn}(\sigma^{k})=\mathrm{sgn}%
(\tilde{\pi}^{k})\mathrm{sgn}(\tilde{\sigma}^{k})$, which follows by a similar
argument as \eqref{eq:identity} above. We get
\begin{align}
&\frac{(n-k)!}{n!}\sum_{A_{k}\ k\mathrm{-set}}\sum_{\pi^{k},\sigma^{k}\in
S_{k}}\mathrm{sgn}(\pi^{k})\mathrm{sgn}(\sigma^{k})\nonumber
|\phi_{(f_{A_{k}}\circ
\pi^{k})(1)}\rangle\langle\phi_{(f_{A_{k}}\circ\sigma^{k})(1)}|\otimes
\cdots\otimes|\phi_{(f_{A_{k}}\circ\pi^{k})(k)}\rangle\langle\phi_{(f_{A_{k}%
}\circ\sigma^{k})(k)}|\\
&  =\frac{1}{\binom{n}{k}}\sum_{A_{k}\ k\mathrm{-set}}\frac{1}{k!}\sum
_{\tilde{\pi}^{k},\tilde{\sigma}^{k}\in\text{Perm}(A_{k})}\mathrm{sgn}%
(\tilde{\pi}^{k})\mathrm{sgn}(\tilde{\sigma}^{k})|\phi_{\tilde{\pi}^{k}%
(i_1)}\rangle\langle\phi_{\tilde{\sigma}^{k}(i_1)}|\otimes\cdots\otimes
|\phi_{\tilde{\pi}^{k}(i_k)}\rangle\langle\phi_{\tilde{\sigma}^{k}(i_k)}|\\
&  =\frac{1}{\binom{n}{k}}\sum_{A_{k}\ k\mathrm{-set}}|\Phi_{A_{k}}%
\rangle\langle\Phi_{A_{k}}|.
\end{align}
This concludes the proof.
\end{proof}

\section{Limitations on approximate two-fold broadcasts}

As mentioned in the main text, our method also gives limitations on approximate two-fold broadcasting. 

Throughout, we restrict to broadcasts which receive as their input state only a single copy of $\sigma$. In particular, we are not in a situation where ``superbroadcasting'' \cite{SB1,SB2} is possible.

\be{thm}
\label{thm:mainBC}
Fix two mixed states $\sigma_1$ and $\sigma_2$. Suppose that the quantum channel
$\Lam_{A\to AB}$ is a simultaneous approximate broadcast of $\sigma_1$ and $\sigma_2$, i.e., that
\beq
\rho^{\mathrm{out}}_{i,A}=\rho^{\mathrm{out}}_{i,B}=\tilde\sigma_{i},
\qquad \rho^{\mathrm{out}}_{i,AB}:=\Lam(\sigma_{i,A})
\eeq
for $i=1,2$. Then
\beq
\label{eq:mainBC}
D(\sigma_1\|\sigma_2)-D(\tilde\sigma_1\|\tilde\sigma_2)\geq 
\De_{\curly{R}}(\tilde\sigma_1,\tilde\sigma_2).
\eeq
where we have introduced the (channel dependent) ``recovery difference''
\beq
\label{eq:commondefn}
\De_\curly{R}(\tilde\sigma_1,\tilde\sigma_2):=
\frac{1}{8} 
\int_\R \|\curly{R}^t_{B,\rho^{\mathrm{out}}_{2,AB}}(\tilde\sigma_{1,A})-\curly{R}^t_{A,\rho^{\mathrm{out}}_{2,AB}}(\tilde\sigma_{1,B})\|_1^2 \ \d\beta(t).
\eeq
which features the probability distribution $\beta(t):=\frac{\pi}{2}(1+\cosh(\pi t))^{-1}$ and the  rotated Petz recovery map defined by
\beq
\label{eq:rPetz}
\curly{R}^t_{A,X}(\cdot):=
X_{AB}^{(1+it)/2}	\l(I_A\otimes X_B^{-(1+it)/2}	(\cdot)X_B^{-(1-it)/2}\r)X_{AB}^{(1-it)/2}.
\eeq
\e{thm}


The proof is given at the end of this appendix. We emphasize that the definition \eqref{eq:commondefn} of the recovery difference $\De_\curly{R}(\tilde\sigma_1,\tilde\sigma_2)$ is independent of $\rho^{\mathrm{out}}_{1,AB}$.
The  rotated Petz recovery map \eqref{eq:rPetz} appears in the strengthening of the monotonicity of relative entropy \cite{Jungeetal}, recalled here as Theorem~\ref{thm:monotonicity} in the appendix. The rotated Petz recovery map is chosen such that the second state is perfectly recovered, i.e.
$$
\curly{R}^t_{B,\rho^{\mathrm{out}}_{2,AB}}(\tilde\sigma_{2,A})=\curly{R}^t_{A,\rho^{\mathrm{out}}_{2,AB}}(\tilde\sigma_{2,B})=\rho^{\mathrm{out}}_{2,AB}.
$$

%
%
%
 One may wonder if the vanishing of the recovery difference implies that $\tilde\sigma_1$ and $\tilde\sigma_2$ commute, i.e., if Theorem~\ref{thm:UBC} is recovered from Theorem~\ref{thm:mainBC}. Assume that $\De_\curly{R}(\tilde\sigma_1,\tilde\sigma_2)=0$. One would like to show that this implies that $\tilde\sigma_1$ and $\tilde\sigma_2$ commute. A natural idea is to follow the proof of Theorem~\ref{thm:UBC} in \cite{KH08}. There, the authors appeal to a condition for equality in the monotonicity of the relative entropy by Ruskai \cite{Ruskai} (see also \cite{Haydenetal,Petz03,Petz86}). It yields (see (11) in \cite{KH08})
\beq
\label{eq:SigmaKH}
(\Sigma_A\otimes I_B) P_{AB}=(I_A\otimes \Sigma_B)P_{AB},\qquad \Sigma:=\log\sigma_1-\log\sigma_2.
\eeq
where $P_{AB}$ projects onto the support of $\rho_{2,AB}^{\mathrm{out}}$. We have

\be{lm}
\label{lm:KH}
If \eqref{eq:SigmaKH} holds, then $\tilde\sigma_1$ and $\tilde\sigma_2$ commute.
\e{lm}
This was observed without proof in \cite{KH08}; for completeness we include the

\be{proof}[Proof of Lemma \ref{lm:KH}]
First, recall our standing assumption that $\ker\tilde\sigma_2\subset \ker\tilde\sigma_1$. It yields that $\tilde\sigma_1\tilde\sigma_2=0=\tilde\sigma_2\tilde\sigma_1$ on $\ker\tilde\sigma_2$
and so it suffices to consider the subspace $X:=(\ker\tilde\sigma_2)^\perp$ in the following.

Fix a vector $\ket{k}\in X$. Then, by the definition of the partial trace, there exists another vector $\ket{l}$ such that 
$$
\ket{k}_A\otimes\ket{l}_B\in (\ker\rho^{\mathrm{out}}_2)^\perp=\mathrm{supp} \rho^{\mathrm{out}}_2.
$$
Hence we have \eqref{eq:Sigma} when acting on $\ket{k}\otimes\ket{l}$, which implies
$
\Sigma \ket{k}=\ket{k}.
$
Since $\ket{k}\in X$ was arbitrary, we see that $\Sigma$ acts as the identity on $X$. Moreover, $X=\mathrm{ran}\tilde\sigma_2$ is an invariant subspace for $\sigma_2$ and so we can find a unitary $U:X\to X$ such that $U^*\tilde
\sigma_2 U=:\Lam$ is diagonal. By definition \eqref{eq:Sigma} of $\Sigma$, it follows that, on $X$,
$$
I_X=\Lam^{-1/2-it/2}	U^*\tilde\sigma_1 U \Lam^{-1/2+it/2}.
$$
Hence, $U^*\tilde\sigma_1 U$ is diagonal as well, implying that $\tilde\sigma_1$ and $\tilde\sigma_2$ commute.
\e{proof}

Contrary to \cite{KH08}, the assumption $\De_\curly{R}(\tilde\sigma_1,\tilde\sigma_2)=0$, by \eqref{eq:commondefn}, yields only the slightly weaker identity
\beq
\label{eq:Sigma}
P_{AB}(\Sigma_A\otimes I_B) P_{AB}=P_{AB}(I_A\otimes \Sigma_B)P_{AB},\qquad \Sigma:=\tilde\sigma_{2}^{-1/2-it/}	\tilde\sigma_{1}\tilde\sigma_{2}^{-1/2+it/2}.
\eeq 
Note the additional projection $P_{AB}$ in \eqref{eq:Sigma} as compared to \eqref{eq:SigmaKH}. It is due to the symmetrical appearance of $\rho_2^{\mathrm{out}}$ in the Petz recovery map \eqref{eq:rPetz}. In the special case that $P_{AB}$ projects onto a subset of the ``diagonal'' $\ket{k}_A\otimes\ket{k}_B$, \eqref{eq:Sigma} holds trivially. In particular, \eqref{eq:Sigma} does \emph{not} imply that $\tilde\sigma_1$ and $\tilde\sigma_2$ commute.\\

Now, if one is intent on recovering the no-broadcasting Theorem \ref{thm:UBC}, one can in fact replace $\De_\curly{R}$ on the right-hand side in \eqref{eq:mainBC} by an alternative expression whose vanishing does imply that $\tilde\sigma_1$ and $\tilde\sigma_2$ commute. This alternative expression is derived from a strengthened monotonicity inequality of Carlen and Lieb \cite{CarlenLieb14} and reads
\beq
\label{eq:Decl}
\begin{aligned}
\De_{CL}(\tilde\sigma_1,\tilde\sigma_2):=&\frac{1}{2}\l\|\sqrt{\rho_{2,AB}^{\mathrm{out}}}-\exp\l(\frac{1}{2}(\log\rho_{2,AB}^{\mathrm{out}}-\log\tilde\sigma_{2,A}+\log\tilde\sigma_{1,A})P_{AB}\r)\r\|_2^2\\
&+\frac{1}{2}\l\|\sqrt{\rho_{2,AB}^{\mathrm{out}}}-\exp\l(\frac{1}{2}(\log\rho_{2,AB}^{\mathrm{out}}-\log\tilde\sigma_{2,B}+\log\tilde\sigma_{1,B})P_{AB}\r)\r\|_2^2
\end{aligned}
\eeq
Using the result of \cite{CarlenLieb14} in the proof of Theorem \ref{thm:mainBC} gives
$$
D(\sigma_1\|\sigma_2)-D(\tilde\sigma_1\|\tilde\sigma_2)\geq \De_{\operatorname{CL}}(\tilde\sigma_1,\tilde\sigma_2),
$$
The vanishing $\De_{CL}(\tilde\sigma_1,\tilde\sigma_2)=0$ implies Ruskai's condition \eqref{eq:SigmaKH} and consequently that $\tilde\sigma_1$ and $\tilde\sigma_2$ commute, i.e.
 \beq
 \label{eq:commute}
\De_{\operatorname{CL}}(\tilde\sigma_1,\tilde\sigma_2)=0\quad\Rightarrow\quad [\tilde\sigma_1,\tilde\sigma_2]=0.
\eeq
However, $\De_{\operatorname{CL}}$ does not appear to have information-theoretic content, while $\De_{\curly{R}}$ features the Petz recovery map.

We close this appendix with the

\be{proof}[Proof of Theorem~\ref{thm:mainBC}.]

The proof is based on the following key estimate.  It is a variant of Theorem~\ref{thm:monotonicity}, which was proved in \cite{Jungeetal}.

 \be{lm}[Key estimate]
\label{lm:main}
Fix two quantum states $\sigma_1$ and $\sigma_2$. For any choice of quantum channel $\Lam_{A\to AB}$, we define
\beq
\rho^{\mathrm{out}}_i:=\Lam(\sigma_{i,A}),\qquad (i=1,2).
\eeq
Let $\beta(t)=\frac{\pi}{2}(1+\cosh(\pi t))^{-1}$.

\be{enumerate}[label=(\roman*)]
\item
We have
\beq
\label{eq:maini'}
D(\sigma_1\|\sigma_2)-D(\rho^{\mathrm{out}}_{1,B}\|\rho^{\mathrm{out}}_{2,B})
\geq
-\int_\R \log F\l(\rho^{\mathrm{out}}_{1,AB}, \curly{R}^t_{A,\rho^{\mathrm{out}}_{2,AB}}(\rho^{\mathrm{out}}_{1,B})	\r) \,\d\beta(t).
\eeq
\beq
\label{eq:maini}
D(\sigma_1\|\sigma_2)-D(\rho^{\mathrm{out}}_{1,A}\|\rho^{\mathrm{out}}_{2,A})
\geq
-\int_\R \log F\l(\rho^{\mathrm{out}}_{1,AB}, \curly{R}^t_{B,\rho^{\mathrm{out}}_{2,AB}}(\rho^{\mathrm{out}}_{1,A})	\r) \,\d\beta(t),
\eeq
where the rotated Petz recovery map $\curly{R}^t_{A,X}$ was defined in \eqref{eq:rPetz}.
\item Suppose that the output state $\rho^{\mathrm{out}}_{i,AB}$ has identical marginals, i.e.\
$$
\rho^{\mathrm{out}}_{i,A}=\rho^{\mathrm{out}}_{i,B}=:\tilde \sigma_i,\qquad (i=1,2).
$$
Then we have
\beq
\label{eq:mainii}
D(\sigma_1\|\sigma_2)-D(\tilde\sigma_1\|\tilde\sigma_2)
\geq
\be{cases}
-\int_\R \log F\l(\rho^{\mathrm{out}}_{1,AB}, \curly{R}^t_{A,\rho^{\mathrm{out}}_{2,AB}}(\tilde\sigma_{1,B})	\r) \,\d\beta(t)\\
-\int_\R \log F\l(\rho^{\mathrm{out}}_{1,AB}, \curly{R}^t_{B,\rho^{\mathrm{out}}_{2,AB}}(\tilde\sigma_{1,A})	\r) \,\d\beta(t).
\e{cases}
\eeq 
\e{enumerate}
\end{lm}

\be{proof}[Proof of Lemma \ref{lm:main}]
The standard monotonicity of quantum relative entropy under quantum channels (without a remainder term) gives
$$
D(\sigma_1\|\sigma_2)\geq D(\Lam(\sigma_1)\|\Lam(\sigma_2))=D(\rho^{\mathrm{out}}_1\| \rho^{\mathrm{out}}_2).
$$
Consider the last expression. When we apply the partial trace over the $A$ subsystem to both states and use Theorem~\ref{thm:monotonicity}, we obtain 
\begin{equation}
D(\rho^{\mathrm{out}}_1\| \rho^{\mathrm{out}}_2)\geq D(\rho^{\mathrm{out}}_{1,B}\|\rho^{\mathrm{out}}_{2,B})
-\int_\R \log F\l(\rho^{\mathrm{out}}_{1,AB}, \curly{R}^t_{\rho^{\mathrm{out}}_{2,AB}}(\rho^{\mathrm{out}}_{1,B})	\r) \, \d\beta(t). \nonumber
\end{equation}
This proves \eqref{eq:maini'} and \eqref{eq:maini} follows by the same argument, only that the $B$ subsystem is traced out now. Statement (ii) is immediate. 
\e{proof}

With Lemma \ref{lm:main} at our disposal, we can now prove Theorem~\ref{thm:mainBC}. We begin by applying Lemma \ref{lm:main} (ii), averaging the two lines in \eqref{eq:mainii}. We get
\begin{equation}
D(\sigma_1\|\sigma_2)-D(\tilde\sigma_1\|\tilde\sigma_2)
\geq
-\frac{1}{2}\int_\R 
\Bigg(\log F\l(\rho^{\mathrm{out}}_{1,AB}, \curly{R}^t_{B,\rho^{\mathrm{out}}_{2,AB}}(\tilde\sigma_{1,A})\r)	
+\log F\l(\rho^{\mathrm{out}}_{1,AB}, \curly{R}^t_{A,\rho^{\mathrm{out}}_{2,AB}}(\tilde\sigma_{1,B})\r)	\Bigg) \, \d\beta(t).\nonumber
\end{equation}
By an elementary estimate and the Fuchs-van de Graaf inequality \cite{FG}, we have for density operators $\omega$ and $\tau$ that
$$
-\log F(\omega,\tau)\geq 1-F(\omega,\tau)\geq \frac{1}{4}\|\omega-\tau\|_1^2.
$$
We apply this to the integrand above, followed by the estimate
$$
\|X-Y\|_1^2+\|X-Z\|_1^2\geq \frac{1}{2}\|Y-Z\|_1^2,
$$
which is a consequence of the triangle inequality and the elementary bound $2ab\leq a^2+b^2$. We conclude
\begin{equation}
D(\sigma_1\|\sigma_2)-D(\tilde\sigma_1\|\tilde\sigma_2)
\geq \frac{1}{8}\int_\R \|\curly{R}^t_{B,\rho^{\mathrm{out}}_{2,AB}}(\tilde\sigma_{1,A})-\curly{R}^t_{A,\rho^{\mathrm{out}}_{2,AB}}(\tilde\sigma_{1,B})\|_1^2 \, \d\beta(t). \nonumber
\end{equation}
This proves Theorem~\ref{thm:mainBC}.
\e{proof}


\begin{thebibliography}{46}%
\makeatletter
\providecommand \@ifxundefined [1]{%
 \@ifx{#1\undefined}
}%
\providecommand \@ifnum [1]{%
 \ifnum #1\expandafter \@firstoftwo
 \else \expandafter \@secondoftwo
 \fi
}%
\providecommand \@ifx [1]{%
 \ifx #1\expandafter \@firstoftwo
 \else \expandafter \@secondoftwo
 \fi
}%
\providecommand \natexlab [1]{#1}%
\providecommand \enquote  [1]{``#1''}%
\providecommand \bibnamefont  [1]{#1}%
\providecommand \bibfnamefont [1]{#1}%
\providecommand \citenamefont [1]{#1}%
\providecommand \href@noop [0]{\@secondoftwo}%
\providecommand \href [0]{\begingroup \@sanitize@url \@href}%
\providecommand \@href[1]{\@@startlink{#1}\@@href}%
\providecommand \@@href[1]{\endgroup#1\@@endlink}%
\providecommand \@sanitize@url [0]{\catcode `\\12\catcode `\$12\catcode
  `\&12\catcode `\#12\catcode `\^12\catcode `\_12\catcode `\%12\relax}%
\providecommand \@@startlink[1]{}%
\providecommand \@@endlink[0]{}%
\providecommand \url  [0]{\begingroup\@sanitize@url \@url }%
\providecommand \@url [1]{\endgroup\@href {#1}{\urlprefix }}%
\providecommand \urlprefix  [0]{URL }%
\providecommand \Eprint [0]{\href }%
\providecommand \doibase [0]{http://dx.doi.org/}%
\providecommand \selectlanguage [0]{\@gobble}%
\providecommand \bibinfo  [0]{\@secondoftwo}%
\providecommand \bibfield  [0]{\@secondoftwo}%
\providecommand \translation [1]{[#1]}%
\providecommand \BibitemOpen [0]{}%
\providecommand \bibitemStop [0]{}%
\providecommand \bibitemNoStop [0]{.\EOS\space}%
\providecommand \EOS [0]{\spacefactor3000\relax}%
\providecommand \BibitemShut  [1]{\csname bibitem#1\endcsname}%
\let\auto@bib@innerbib\@empty
\bibitem [{\citenamefont {Dieks}(1982)}]{Dieks}%
  \BibitemOpen
  \bibfield  {author} {\bibinfo {author} {\bibfnamefont {D.}~\bibnamefont
  {Dieks}},\ }\href@noop {} {\bibfield  {journal} {\bibinfo  {journal} {Physics
  Letters A}\ }\textbf {\bibinfo {volume} {92}},\ \bibinfo {pages} {271 }
  (\bibinfo {year} {1982})}\BibitemShut {NoStop}%
\bibitem [{\citenamefont {Wootters}\ and\ \citenamefont
  {Zurek}(1982)}]{Wootters}%
  \BibitemOpen
  \bibfield  {author} {\bibinfo {author} {\bibfnamefont {W.}~\bibnamefont
  {Wootters}}\ and\ \bibinfo {author} {\bibfnamefont {W.}~\bibnamefont
  {Zurek}},\ }\href@noop {} {\bibfield  {journal} {\bibinfo  {journal}
  {Nature}\ }\textbf {\bibinfo {volume} {299}},\ \bibinfo {pages} {802–803}
  (\bibinfo {year} {1982})}\BibitemShut {NoStop}%
\bibitem [{\citenamefont {Barnum}\ \emph {et~al.}(1996)\citenamefont {Barnum},
  \citenamefont {Caves}, \citenamefont {Fuchs}, \citenamefont {Jozsa},\ and\
  \citenamefont {Schumacher}}]{Barnumetal96}%
  \BibitemOpen
  \bibfield  {author} {\bibinfo {author} {\bibfnamefont {H.}~\bibnamefont
  {Barnum}}, \bibinfo {author} {\bibfnamefont {C.~M.}\ \bibnamefont {Caves}},
  \bibinfo {author} {\bibfnamefont {C.~A.}\ \bibnamefont {Fuchs}}, \bibinfo
  {author} {\bibfnamefont {R.}~\bibnamefont {Jozsa}}, \ and\ \bibinfo {author}
  {\bibfnamefont {B.}~\bibnamefont {Schumacher}},\ }\href@noop {} {\bibfield
  {journal} {\bibinfo  {journal} {Phys. Rev. Lett.}\ }\textbf {\bibinfo
  {volume} {76}},\ \bibinfo {pages} {2818} (\bibinfo {year}
  {1996})}\BibitemShut {NoStop}%
\bibitem [{\citenamefont {Knill}\ and\ \citenamefont {Laflamme}(1997)}]{Knill}%
  \BibitemOpen
  \bibfield  {author} {\bibinfo {author} {\bibfnamefont {E.}~\bibnamefont
  {Knill}}\ and\ \bibinfo {author} {\bibfnamefont {R.}~\bibnamefont
  {Laflamme}},\ }\href@noop {} {\bibfield  {journal} {\bibinfo  {journal}
  {Phys. Rev. A}\ }\textbf {\bibinfo {volume} {55}},\ \bibinfo {pages} {900}
  (\bibinfo {year} {1997})}\BibitemShut {NoStop}%
\bibitem [{\citenamefont {Mandayam}\ and\ \citenamefont {Ng}(2012)}]{Manda}%
  \BibitemOpen
  \bibfield  {author} {\bibinfo {author} {\bibfnamefont {P.}~\bibnamefont
  {Mandayam}}\ and\ \bibinfo {author} {\bibfnamefont {H.~K.}\ \bibnamefont
  {Ng}},\ }\href@noop {} {\bibfield  {journal} {\bibinfo  {journal} {Phys. Rev.
  A}\ }\textbf {\bibinfo {volume} {86}},\ \bibinfo {pages} {012335} (\bibinfo
  {year} {2012})}\BibitemShut {NoStop}%
\bibitem [{\citenamefont {Shor}(1996)}]{Shor}%
  \BibitemOpen
  \bibfield  {author} {\bibinfo {author} {\bibfnamefont {P.~W.}\ \bibnamefont
  {Shor}},\ }in\ \href@noop {} {\emph {\bibinfo {booktitle} {Proceedings of the
  37th Annual Symposium on Foundations of Computer Science}}},\ \bibinfo
  {series and number} {FOCS '96}\ (\bibinfo  {publisher} {IEEE Computer
  Society},\ \bibinfo {address} {Washington, DC, USA},\ \bibinfo {year}
  {1996})\ pp.\ \bibinfo {pages} {56--}\BibitemShut {NoStop}%
\bibitem [{\citenamefont {Bu\ifmmode~\check{z}\else \v{z}\fi{}ek}\ and\
  \citenamefont {Hillery}(1996)}]{Bu}%
  \BibitemOpen
  \bibfield  {author} {\bibinfo {author} {\bibfnamefont {V.}~\bibnamefont
  {Bu\ifmmode~\check{z}\else \v{z}\fi{}ek}}\ and\ \bibinfo {author}
  {\bibfnamefont {M.}~\bibnamefont {Hillery}},\ }\href@noop {} {\bibfield
  {journal} {\bibinfo  {journal} {Phys. Rev. A}\ }\textbf {\bibinfo {volume}
  {54}},\ \bibinfo {pages} {1844} (\bibinfo {year} {1996})}\BibitemShut
  {NoStop}%
\bibitem [{\citenamefont {Gisin}\ and\ \citenamefont {Massar}(1997)}]{Gisin}%
  \BibitemOpen
  \bibfield  {author} {\bibinfo {author} {\bibfnamefont {N.}~\bibnamefont
  {Gisin}}\ and\ \bibinfo {author} {\bibfnamefont {S.}~\bibnamefont {Massar}},\
  }\href@noop {} {\bibfield  {journal} {\bibinfo  {journal} {Phys. Rev. Lett.}\
  }\textbf {\bibinfo {volume} {79}},\ \bibinfo {pages} {2153} (\bibinfo {year}
  {1997})}\BibitemShut {NoStop}%
\bibitem [{\citenamefont {Werner}(1998)}]{PhysRevA.58.1827}%
  \BibitemOpen
  \bibfield  {author} {\bibinfo {author} {\bibfnamefont {R.~F.}\ \bibnamefont
  {Werner}},\ }\href@noop {} {\bibfield  {journal} {\bibinfo  {journal} {Phys.
  Rev. A}\ }\textbf {\bibinfo {volume} {58}},\ \bibinfo {pages} {1827}
  (\bibinfo {year} {1998})}\BibitemShut {NoStop}%
\bibitem [{\citenamefont {Allahverdyan}\ and\ \citenamefont
  {Hovhannisyan}(2010)}]{Karen}%
  \BibitemOpen
  \bibfield  {author} {\bibinfo {author} {\bibfnamefont {A.~E.}\ \bibnamefont
  {Allahverdyan}}\ and\ \bibinfo {author} {\bibfnamefont {K.~V.}\ \bibnamefont
  {Hovhannisyan}},\ }\href {\doibase 10.1103/PhysRevA.81.012312} {\bibfield
  {journal} {\bibinfo  {journal} {Phys. Rev. A}\ }\textbf {\bibinfo {volume}
  {81}},\ \bibinfo {pages} {012312} (\bibinfo {year} {2010})}\BibitemShut
  {NoStop}%
\bibitem [{\citenamefont {Keyl}\ and\ \citenamefont {Werner}(1999)}]{KW99}%
  \BibitemOpen
  \bibfield  {author} {\bibinfo {author} {\bibfnamefont {M.}~\bibnamefont
  {Keyl}}\ and\ \bibinfo {author} {\bibfnamefont {R.~F.}\ \bibnamefont
  {Werner}},\ }\href@noop {} {\bibfield  {journal} {\bibinfo  {journal}
  {Journal of Mathematical Physics}\ }\textbf {\bibinfo {volume} {40}},\
  \bibinfo {pages} {3283} (\bibinfo {year} {1999})}\BibitemShut {NoStop}%
\bibitem [{\citenamefont {Lamas-Linares}\ \emph {et~al.}(2002)\citenamefont
  {Lamas-Linares}, \citenamefont {Simon}, \citenamefont {Howell},\ and\
  \citenamefont {Bouwmeester}}]{Lamas}%
  \BibitemOpen
  \bibfield  {author} {\bibinfo {author} {\bibfnamefont {A.}~\bibnamefont
  {Lamas-Linares}}, \bibinfo {author} {\bibfnamefont {C.}~\bibnamefont
  {Simon}}, \bibinfo {author} {\bibfnamefont {J.~C.}\ \bibnamefont {Howell}}, \
  and\ \bibinfo {author} {\bibfnamefont {D.}~\bibnamefont {Bouwmeester}},\
  }\href@noop {} {\bibfield  {journal} {\bibinfo  {journal} {Science}\ }
  (\bibinfo {year} {2002})}\BibitemShut {NoStop}%
\bibitem [{\citenamefont {Scarani}\ \emph {et~al.}(2005)\citenamefont
  {Scarani}, \citenamefont {Iblisdir}, \citenamefont {Gisin},\ and\
  \citenamefont {Ac\'{\i}n}}]{Scarani}%
  \BibitemOpen
  \bibfield  {author} {\bibinfo {author} {\bibfnamefont {V.}~\bibnamefont
  {Scarani}}, \bibinfo {author} {\bibfnamefont {S.}~\bibnamefont {Iblisdir}},
  \bibinfo {author} {\bibfnamefont {N.}~\bibnamefont {Gisin}}, \ and\ \bibinfo
  {author} {\bibfnamefont {A.}~\bibnamefont {Ac\'{\i}n}},\ }\href@noop {}
  {\bibfield  {journal} {\bibinfo  {journal} {Rev. Mod. Phys.}\ }\textbf
  {\bibinfo {volume} {77}},\ \bibinfo {pages} {1225} (\bibinfo {year}
  {2005})}\BibitemShut {NoStop}%
\bibitem [{\citenamefont {Fan}\ \emph {et~al.}(2014)\citenamefont {Fan},
  \citenamefont {Wang}, \citenamefont {Jing}, \citenamefont {Yue},
  \citenamefont {Shi}, \citenamefont {Zhang},\ and\ \citenamefont
  {Mu}}]{Review}%
  \BibitemOpen
  \bibfield  {author} {\bibinfo {author} {\bibfnamefont {H.}~\bibnamefont
  {Fan}}, \bibinfo {author} {\bibfnamefont {Y.-N.}\ \bibnamefont {Wang}},
  \bibinfo {author} {\bibfnamefont {L.}~\bibnamefont {Jing}}, \bibinfo {author}
  {\bibfnamefont {J.-D.}\ \bibnamefont {Yue}}, \bibinfo {author} {\bibfnamefont
  {H.-D.}\ \bibnamefont {Shi}}, \bibinfo {author} {\bibfnamefont {Y.-L.}\
  \bibnamefont {Zhang}}, \ and\ \bibinfo {author} {\bibfnamefont {L.-Z.}\
  \bibnamefont {Mu}},\ }\href@noop {} {\bibfield  {journal} {\bibinfo
  {journal} {Physics Reports}\ }\textbf {\bibinfo {volume} {544}},\ \bibinfo
  {pages} {241 } (\bibinfo {year} {2014})}\BibitemShut {NoStop}%
\bibitem [{\citenamefont {Zhu}\ and\ \citenamefont {Ye}(2015)}]{Zhu}%
  \BibitemOpen
  \bibfield  {author} {\bibinfo {author} {\bibfnamefont {M.-Z.}\ \bibnamefont
  {Zhu}}\ and\ \bibinfo {author} {\bibfnamefont {L.}~\bibnamefont {Ye}},\
  }\href@noop {} {\bibfield  {journal} {\bibinfo  {journal} {Phys. Rev. A}\
  }\textbf {\bibinfo {volume} {91}},\ \bibinfo {pages} {042319} (\bibinfo
  {year} {2015})}\BibitemShut {NoStop}%
\bibitem [{\citenamefont {Chatterjee}\ \emph {et~al.}(2016)\citenamefont
  {Chatterjee}, \citenamefont {Sazim},\ and\ \citenamefont
  {Chakrabarty}}]{Chatterjeeetal}%
  \BibitemOpen
  \bibfield  {author} {\bibinfo {author} {\bibfnamefont {S.}~\bibnamefont
  {Chatterjee}}, \bibinfo {author} {\bibfnamefont {S.}~\bibnamefont {Sazim}}, \
  and\ \bibinfo {author} {\bibfnamefont {I.}~\bibnamefont {Chakrabarty}},\
  }\href@noop {} {\bibfield  {journal} {\bibinfo  {journal} {Phys. Rev. A}\
  }\textbf {\bibinfo {volume} {93}},\ \bibinfo {pages} {042309} (\bibinfo
  {year} {2016})}\BibitemShut {NoStop}%
\bibitem [{\citenamefont {Kalev}\ and\ \citenamefont {Hen}(2008)}]{KH08}%
  \BibitemOpen
  \bibfield  {author} {\bibinfo {author} {\bibfnamefont {A.}~\bibnamefont
  {Kalev}}\ and\ \bibinfo {author} {\bibfnamefont {I.}~\bibnamefont {Hen}},\
  }\href@noop {} {\bibfield  {journal} {\bibinfo  {journal} {Phys. Rev. Lett.}\
  }\textbf {\bibinfo {volume} {100}},\ \bibinfo {pages} {210502} (\bibinfo
  {year} {2008})}\BibitemShut {NoStop}%
\bibitem [{\citenamefont {Lindblad}(1975)}]{Lindblad1975}%
  \BibitemOpen
  \bibfield  {author} {\bibinfo {author} {\bibfnamefont {G.}~\bibnamefont
  {Lindblad}},\ }\href@noop {} {\bibfield  {journal} {\bibinfo  {journal}
  {Comm. Math. Phys.}\ }\textbf {\bibinfo {volume} {40}},\ \bibinfo {pages}
  {147} (\bibinfo {year} {1975})}\BibitemShut {NoStop}%
\bibitem [{\citenamefont {Uhlmann}(1977)}]{U77}%
  \BibitemOpen
  \bibfield  {author} {\bibinfo {author} {\bibfnamefont {A.}~\bibnamefont
  {Uhlmann}},\ }\href@noop {} {\bibfield  {journal} {\bibinfo  {journal} {Comm.
  Math. Phys.}\ }\textbf {\bibinfo {volume} {54}},\ \bibinfo {pages} {21}
  (\bibinfo {year} {1977})}\BibitemShut {NoStop}%
\bibitem [{\citenamefont {Fawzi}\ and\ \citenamefont {Renner}(2015)}]{FR}%
  \BibitemOpen
  \bibfield  {author} {\bibinfo {author} {\bibfnamefont {O.}~\bibnamefont
  {Fawzi}}\ and\ \bibinfo {author} {\bibfnamefont {R.}~\bibnamefont {Renner}},\
  }\href@noop {} {\bibfield  {journal} {\bibinfo  {journal} {Comm. Math.
  Phys.}\ }\textbf {\bibinfo {volume} {340}},\ \bibinfo {pages} {575} (\bibinfo
  {year} {2015})}\BibitemShut {NoStop}%
\bibitem [{\citenamefont {Berta}\ \emph {et~al.}(2015)\citenamefont {Berta},
  \citenamefont {Lemm},\ and\ \citenamefont {Wilde}}]{BLW}%
  \BibitemOpen
  \bibfield  {author} {\bibinfo {author} {\bibfnamefont {M.}~\bibnamefont
  {Berta}}, \bibinfo {author} {\bibfnamefont {M.}~\bibnamefont {Lemm}}, \ and\
  \bibinfo {author} {\bibfnamefont {M.~M.}\ \bibnamefont {Wilde}},\ }\href@noop
  {} {\bibfield  {journal} {\bibinfo  {journal} {Quantum Info. Comput.}\
  }\textbf {\bibinfo {volume} {15}},\ \bibinfo {pages} {1333} (\bibinfo {year}
  {2015})}\BibitemShut {NoStop}%
\bibitem [{\citenamefont {Sutter}\ \emph {et~al.}(2016)\citenamefont {Sutter},
  \citenamefont {Fawzi},\ and\ \citenamefont {Renner}}]{SFR}%
  \BibitemOpen
  \bibfield  {author} {\bibinfo {author} {\bibfnamefont {D.}~\bibnamefont
  {Sutter}}, \bibinfo {author} {\bibfnamefont {O.}~\bibnamefont {Fawzi}}, \
  and\ \bibinfo {author} {\bibfnamefont {R.}~\bibnamefont {Renner}},\
  }\href@noop {} {\bibfield  {journal} {\bibinfo  {journal} {Proc. R. Soc. A.}\
  }\textbf {\bibinfo {volume} {472}},\ \bibinfo {pages} {20150623} (\bibinfo
  {year} {2016})}\BibitemShut {NoStop}%
\bibitem [{\citenamefont {Wilde}(2015)}]{Wilde}%
  \BibitemOpen
  \bibfield  {author} {\bibinfo {author} {\bibfnamefont {M.~M.}\ \bibnamefont
  {Wilde}},\ }\href@noop {} {\bibfield  {journal} {\bibinfo  {journal} {Proc.
  R. Soc. A}\ }\textbf {\bibinfo {volume} {471}},\ \bibinfo {pages} {20150338}
  (\bibinfo {year} {2015})}\BibitemShut {NoStop}%
\bibitem [{\citenamefont {Junge}\ \emph {et~al.}()\citenamefont {Junge},
  \citenamefont {Renner}, \citenamefont {Sutter}, \citenamefont {Wilde},\ and\
  \citenamefont {Winter}}]{Jungeetal}%
  \BibitemOpen
  \bibfield  {author} {\bibinfo {author} {\bibfnamefont {M.}~\bibnamefont
  {Junge}}, \bibinfo {author} {\bibfnamefont {R.}~\bibnamefont {Renner}},
  \bibinfo {author} {\bibfnamefont {D.}~\bibnamefont {Sutter}}, \bibinfo
  {author} {\bibfnamefont {M.~M.}\ \bibnamefont {Wilde}}, \ and\ \bibinfo
  {author} {\bibfnamefont {A.}~\bibnamefont {Winter}},\ }\href@noop {}
  {}\bibinfo {howpublished} {arXiv:1509.07127}\BibitemShut {NoStop}%
\bibitem [{\citenamefont {Sutter}\ \emph {et~al.}()\citenamefont {Sutter},
  \citenamefont {Berta},\ and\ \citenamefont {Tomamichel}}]{Sutteretal}%
  \BibitemOpen
  \bibfield  {author} {\bibinfo {author} {\bibfnamefont {D.}~\bibnamefont
  {Sutter}}, \bibinfo {author} {\bibfnamefont {M.}~\bibnamefont {Berta}}, \
  and\ \bibinfo {author} {\bibfnamefont {M.}~\bibnamefont {Tomamichel}},\
  }\href@noop {} {}\bibinfo {howpublished} {arXiv:1604.03023}\BibitemShut
  {NoStop}%
\bibitem [{\citenamefont {Yang}\ \emph {et~al.}(2016)\citenamefont {Yang},
  \citenamefont {Chiribella},\ and\ \citenamefont {Hayashi}}]{Giulio}%
  \BibitemOpen
  \bibfield  {author} {\bibinfo {author} {\bibfnamefont {Y.}~\bibnamefont
  {Yang}}, \bibinfo {author} {\bibfnamefont {G.}~\bibnamefont {Chiribella}}, \
  and\ \bibinfo {author} {\bibfnamefont {M.}~\bibnamefont {Hayashi}},\ }\href
  {\doibase 10.1103/PhysRevLett.117.090502} {\bibfield  {journal} {\bibinfo
  {journal} {Phys. Rev. Lett.}\ }\textbf {\bibinfo {volume} {117}},\ \bibinfo
  {pages} {090502} (\bibinfo {year} {2016})}\BibitemShut {NoStop}%
\bibitem [{\citenamefont {Buscemi}\ \emph {et~al.}(2016)\citenamefont
  {Buscemi}, \citenamefont {Das},\ and\ \citenamefont
  {Wilde}}]{PhysRevA.93.062314}%
  \BibitemOpen
  \bibfield  {author} {\bibinfo {author} {\bibfnamefont {F.}~\bibnamefont
  {Buscemi}}, \bibinfo {author} {\bibfnamefont {S.}~\bibnamefont {Das}}, \ and\
  \bibinfo {author} {\bibfnamefont {M.~M.}\ \bibnamefont {Wilde}},\ }\href@noop
  {} {\bibfield  {journal} {\bibinfo  {journal} {Phys. Rev. A}\ }\textbf
  {\bibinfo {volume} {93}},\ \bibinfo {pages} {062314} (\bibinfo {year}
  {2016})}\BibitemShut {NoStop}%
\bibitem [{\citenamefont {Lindblad}(1999)}]{Lindblad99}%
  \BibitemOpen
  \bibfield  {author} {\bibinfo {author} {\bibfnamefont {G.}~\bibnamefont
  {Lindblad}},\ }\href@noop {} {\bibfield  {journal} {\bibinfo  {journal}
  {Lett. Math. Phys.}\ }\textbf {\bibinfo {volume} {47}},\ \bibinfo {pages}
  {189} (\bibinfo {year} {1999})}\BibitemShut {NoStop}%
\bibitem [{\citenamefont {Leifer}(2006)}]{Leifer}%
  \BibitemOpen
  \bibfield  {author} {\bibinfo {author} {\bibfnamefont {M.~S.}\ \bibnamefont
  {Leifer}},\ }\href@noop {} {\bibfield  {journal} {\bibinfo  {journal} {Phys.
  Rev. A}\ }\textbf {\bibinfo {volume} {74}},\ \bibinfo {pages} {042310}
  (\bibinfo {year} {2006})}\BibitemShut {NoStop}%
\bibitem [{\citenamefont {Barnum}\ \emph {et~al.}(2007)\citenamefont {Barnum},
  \citenamefont {Barrett}, \citenamefont {Leifer},\ and\ \citenamefont
  {Wilce}}]{Barnumetal07}%
  \BibitemOpen
  \bibfield  {author} {\bibinfo {author} {\bibfnamefont {H.}~\bibnamefont
  {Barnum}}, \bibinfo {author} {\bibfnamefont {J.}~\bibnamefont {Barrett}},
  \bibinfo {author} {\bibfnamefont {M.}~\bibnamefont {Leifer}}, \ and\ \bibinfo
  {author} {\bibfnamefont {A.}~\bibnamefont {Wilce}},\ }\href@noop {}
  {\bibfield  {journal} {\bibinfo  {journal} {Phys. Rev. Lett.}\ }\textbf
  {\bibinfo {volume} {99}},\ \bibinfo {pages} {240501} (\bibinfo {year}
  {2007})}\BibitemShut {NoStop}%
\bibitem [{\citenamefont {Piani}\ \emph {et~al.}(2008)\citenamefont {Piani},
  \citenamefont {Horodecki},\ and\ \citenamefont {Horodecki}}]{Pianietal}%
  \BibitemOpen
  \bibfield  {author} {\bibinfo {author} {\bibfnamefont {M.}~\bibnamefont
  {Piani}}, \bibinfo {author} {\bibfnamefont {P.}~\bibnamefont {Horodecki}}, \
  and\ \bibinfo {author} {\bibfnamefont {R.}~\bibnamefont {Horodecki}},\
  }\href@noop {} {\bibfield  {journal} {\bibinfo  {journal} {Phys. Rev. Lett.}\
  }\textbf {\bibinfo {volume} {100}},\ \bibinfo {pages} {090502} (\bibinfo
  {year} {2008})}\BibitemShut {NoStop}%
\bibitem [{\citenamefont {Piani}()}]{Piani16}%
  \BibitemOpen
  \bibfield  {author} {\bibinfo {author} {\bibfnamefont {M.}~\bibnamefont
  {Piani}},\ }\href@noop {} {}\bibinfo {howpublished}
  {arXiv:1608.02650}\BibitemShut {NoStop}%
\bibitem [{\citenamefont {Umegaki}(1962)}]{U62}%
  \BibitemOpen
  \bibfield  {author} {\bibinfo {author} {\bibfnamefont {H.}~\bibnamefont
  {Umegaki}},\ }\href@noop {} {\bibfield  {journal} {\bibinfo  {journal} {Kodai
  Math. Seminar Reports}\ }\textbf {\bibinfo {volume} {14}},\ \bibinfo {pages}
  {59} (\bibinfo {year} {1962})}\BibitemShut {NoStop}%
\bibitem [{\citenamefont {Uhlmann}(1976)}]{U73}%
  \BibitemOpen
  \bibfield  {author} {\bibinfo {author} {\bibfnamefont {A.}~\bibnamefont
  {Uhlmann}},\ }\href@noop {} {\bibfield  {journal} {\bibinfo  {journal}
  {Reports Math. Phys.}\ }\textbf {\bibinfo {volume} {9}},\ \bibinfo {pages}
  {273} (\bibinfo {year} {1976})}\BibitemShut {NoStop}%
  \bibitem [{\citenamefont {LemmWilde}(2017)}]{LW16supmat}%
  \BibitemOpen
  \bibfield  {author} {\bibinfo {author} {\bibfnamefont {M.}~\bibnamefont
  {Lemm}},\ {\bibfnamefont {M.M.}~\bibnamefont
  {Wilde, }}}\href@noop {} {\bibfield  {journal} {\bibinfo  {journal}
  {Supplementary Material}\ } (\bibinfo {year} {2017})}\BibitemShut {NoStop}%
\bibitem [{\citenamefont {Ohya}\ and\ \citenamefont {Petz}(1993)}]{OP93}%
  \BibitemOpen
  \bibfield  {author} {\bibinfo {author} {\bibfnamefont {M.}~\bibnamefont
  {Ohya}}\ and\ \bibinfo {author} {\bibfnamefont {D.}~\bibnamefont {Petz}},\
  }\href@noop {} {\emph {\bibinfo {title} {Quantum Entropy and Its Use}}}\
  (\bibinfo  {publisher} {Springer},\ \bibinfo {year} {1993})\BibitemShut
  {NoStop}%
\bibitem [{\citenamefont {Harrow}(2013)}]{H13}%
  \BibitemOpen
  \bibfield  {author} {\bibinfo {author} {\bibfnamefont {A.~W.}\ \bibnamefont
  {Harrow}},\ }\href@noop {} {\  (\bibinfo {year} {2013})},\ \bibinfo {note}
  {arXiv:1308.6595}\BibitemShut {NoStop}%
\bibitem [{\citenamefont {Bradler}(2011)}]{5961814}%
  \BibitemOpen
  \bibfield  {author} {\bibinfo {author} {\bibfnamefont {K.}~\bibnamefont
  {Bradler}},\ }\href@noop {} {\bibfield  {journal} {\bibinfo  {journal} {IEEE
  Transactions on Information Theory}\ }\textbf {\bibinfo {volume} {57}},\
  \bibinfo {pages} {5497} (\bibinfo {year} {2011})}\BibitemShut {NoStop}%
\bibitem [{\citenamefont {Marvian}\ and\ \citenamefont {Lloyd}(2016)}]{ML16}%
  \BibitemOpen
  \bibfield  {author} {\bibinfo {author} {\bibfnamefont {I.}~\bibnamefont
  {Marvian}}\ and\ \bibinfo {author} {\bibfnamefont {S.}~\bibnamefont
  {Lloyd}},\ }\href@noop {} {\  (\bibinfo {year} {2016})}\BibitemShut {NoStop}%
\bibitem [{\citenamefont {D'Ariano}\ \emph {et~al.}(2005)\citenamefont
  {D'Ariano}, \citenamefont {Macchiavello},\ and\ \citenamefont
  {Perinotti}}]{SB1}%
  \BibitemOpen
  \bibfield  {author} {\bibinfo {author} {\bibfnamefont {G.~M.}\ \bibnamefont
  {D'Ariano}}, \bibinfo {author} {\bibfnamefont {C.}~\bibnamefont
  {Macchiavello}}, \ and\ \bibinfo {author} {\bibfnamefont {P.}~\bibnamefont
  {Perinotti}},\ }\href {\doibase 10.1103/PhysRevLett.95.060503} {\bibfield
  {journal} {\bibinfo  {journal} {Phys. Rev. Lett.}\ }\textbf {\bibinfo
  {volume} {95}},\ \bibinfo {pages} {060503} (\bibinfo {year}
  {2005})}\BibitemShut {NoStop}%
\bibitem [{\citenamefont {Buscemi}\ \emph {et~al.}(2006)\citenamefont
  {Buscemi}, \citenamefont {D'Ariano}, \citenamefont {Macchiavello},\ and\
  \citenamefont {Perinotti}}]{SB2}%
  \BibitemOpen
  \bibfield  {author} {\bibinfo {author} {\bibfnamefont {F.}~\bibnamefont
  {Buscemi}}, \bibinfo {author} {\bibfnamefont {G.~M.}\ \bibnamefont
  {D'Ariano}}, \bibinfo {author} {\bibfnamefont {C.}~\bibnamefont
  {Macchiavello}}, \ and\ \bibinfo {author} {\bibfnamefont {P.}~\bibnamefont
  {Perinotti}},\ }\href {\doibase 10.1103/PhysRevA.74.042309} {\bibfield
  {journal} {\bibinfo  {journal} {Phys. Rev. A}\ }\textbf {\bibinfo {volume}
  {74}},\ \bibinfo {pages} {042309} (\bibinfo {year} {2006})}\BibitemShut
  {NoStop}%
\bibitem [{\citenamefont {Ruskai}(2002)}]{Ruskai}%
  \BibitemOpen
  \bibfield  {author} {\bibinfo {author} {\bibfnamefont {M.~B.}\ \bibnamefont
  {Ruskai}},\ }\href@noop {} {\bibfield  {journal} {\bibinfo  {journal} {J.
  Math. Phys.}\ }\textbf {\bibinfo {volume} {43}},\ \bibinfo {pages} {4358}
  (\bibinfo {year} {2002})}\BibitemShut {NoStop}%
\bibitem [{\citenamefont {Hayden}\ \emph {et~al.}(2004)\citenamefont {Hayden},
  \citenamefont {Jozsa}, \citenamefont {Petz},\ and\ \citenamefont
  {Winter}}]{Haydenetal}%
  \BibitemOpen
  \bibfield  {author} {\bibinfo {author} {\bibfnamefont {P.}~\bibnamefont
  {Hayden}}, \bibinfo {author} {\bibfnamefont {R.}~\bibnamefont {Jozsa}},
  \bibinfo {author} {\bibfnamefont {D.}~\bibnamefont {Petz}}, \ and\ \bibinfo
  {author} {\bibfnamefont {A.}~\bibnamefont {Winter}},\ }\href@noop {}
  {\bibfield  {journal} {\bibinfo  {journal} {Comm. Math. Phys.}\ }\textbf
  {\bibinfo {volume} {246}},\ \bibinfo {pages} {359} (\bibinfo {year}
  {2004})}\BibitemShut {NoStop}%
\bibitem [{\citenamefont {Petz}(2003)}]{Petz03}%
  \BibitemOpen
  \bibfield  {author} {\bibinfo {author} {\bibfnamefont {D.}~\bibnamefont
  {Petz}},\ }\href@noop {} {\bibfield  {journal} {\bibinfo  {journal} {Rev.
  Math. Phys.}\ }\textbf {\bibinfo {volume} {15}},\ \bibinfo {pages} {79}
  (\bibinfo {year} {2003})}\BibitemShut {NoStop}%
\bibitem [{\citenamefont {Petz}(1986)}]{Petz86}%
  \BibitemOpen
  \bibfield  {author} {\bibinfo {author} {\bibfnamefont {D.}~\bibnamefont
  {Petz}},\ }\href@noop {} {\bibfield  {journal} {\bibinfo  {journal} {Comm.
  Math. Phys.}\ }\textbf {\bibinfo {volume} {105}},\ \bibinfo {pages} {123}
  (\bibinfo {year} {1986})}\BibitemShut {NoStop}%
\bibitem [{\citenamefont {Carlen}\ and\ \citenamefont
  {Lieb}(2014)}]{CarlenLieb14}%
  \BibitemOpen
  \bibfield  {author} {\bibinfo {author} {\bibfnamefont {E.~A.}\ \bibnamefont
  {Carlen}}\ and\ \bibinfo {author} {\bibfnamefont {E.~H.}\ \bibnamefont
  {Lieb}},\ }\href@noop {} {\bibfield  {journal} {\bibinfo  {journal} {J. Math.
  Phys.}\ }\textbf {\bibinfo {volume} {55}},\ \bibinfo {eid} {042201} (\bibinfo
  {year} {2014})}\BibitemShut {NoStop}%
\bibitem [{\citenamefont {Fuchs}\ and\ \citenamefont {van~de
  Graaf}(1998)}]{FG}%
  \BibitemOpen
  \bibfield  {author} {\bibinfo {author} {\bibfnamefont {C.~A.}\ \bibnamefont
  {Fuchs}}\ and\ \bibinfo {author} {\bibfnamefont {J.}~\bibnamefont {van~de
  Graaf}},\ }\href@noop {} {\bibfield  {journal} {\bibinfo  {journal} {IEEE
  Transactions on Information Theory}\ }\textbf {\bibinfo {volume} {45}},\
  \bibinfo {pages} {1216} (\bibinfo {year} {1998})}\BibitemShut {NoStop}%
\end{thebibliography}
\end{document}